\documentclass[11pt, a4paper, oneside, reqno]{amsart}

\usepackage{xcolor}

\definecolor{darkblue}{rgb}{0.0, 0.0, 0.45}
\definecolor{lightblue}{rgb}{0.94, 0.97, 1.0}  
\definecolor{lightblue2}{rgb}{0.68, 0.85, 0.9}
\definecolor{lightcyan}{rgb}{0.88, 1.0, 1.0}
\definecolor{palepink}{rgb}{0.98, 0.85, 0.87}
\definecolor{Mahogany}{rgb}{0.75, 0.25, 0.0} 
\definecolor{ForestGreen}{rgb}{0.13, 0.55, 0.13}  

\usepackage[colorlinks	= true,
raiselinks	= true,
linkcolor	= darkblue, 
citecolor	= Mahogany,  
urlcolor	= ForestGreen,
pdfauthor	= {Hossein Gorji},
pdftitle	= {},
pdfkeywords	= {},
pdfsubject	= {},
plainpages	= false]{hyperref}

\usepackage{amsmath, amsthm, amssymb, amsfonts}
\usepackage{mathtools, mathrsfs}
\usepackage{enumerate, enumitem}
\usepackage{dsfont}
\usepackage[amssymb, thickqspace]{SIunits}
\usepackage{fancyhdr,mdframed,nicefrac}
\usepackage{epsfig}
\usepackage{graphicx}
\usepackage{float}
\usepackage{caption}
\usepackage{subcaption}
\usepackage{courier}
\usepackage{multirow}
\usepackage{bigstrut}

\allowdisplaybreaks
\date{\today}
\makeatletter
\def\@settitle{\begin{center}%
		\baselineskip14\p@\relax
		\normalfont\LARGE\scshape\bfseries
		\@title
	\end{center}%
}

\def\@setauthors{%
  \begingroup
  \def\thanks{\protect\thanks@warning}%
  \trivlist
  \centering\footnotesize \@topsep30\p@\relax
  \advance\@topsep by -\baselineskip
  \item\relax
  \author@andify\authors
  \def\\{\protect\linebreak}%
  \authors%
  \ifx\@empty\contribs
  \else
    ,\penalty-3 \space \@setcontribs
    \@closetoccontribs
  \fi
  \endtrivlist
  \endgroup
}

\makeatother
\makeatletter

\def\subsection{\@startsection{subsection}{2}%
	\z@{.5\linespacing\@plus.7\linespacing}{.5\linespacing}%
	{\normalfont\large\bfseries}}

\def\subsubsection{\@startsection{subsubsection}{3}%
	\z@{.5\linespacing\@plus.7\linespacing}{.5\linespacing}%
	{\normalfont\itshape}}

\usepackage{multirow}
\usepackage{bigstrut}
\usepackage{wrapfig}
\usepackage{caption}

\newtheorem{theorem}{Theorem}[section]

\newtheorem{corollary}[theorem]{Corollary}
\newtheorem{proposition}[theorem]{Proposition}

\usepackage{graphicx}
\usepackage{dcolumn}
\usepackage{bm}

\usepackage[utf8]{inputenc}
\usepackage[T1]{fontenc}
\usepackage{etoolbox}
\usepackage{xcolor}

\usepackage{booktabs}
\usepackage{multirow} 
\usepackage{mathtools}
\usepackage{amsmath, amsthm, amssymb, amsfonts}
\usepackage[english]{babel}
\usepackage{comment}
\usepackage{algorithm}
\usepackage{algpseudocode}
\usepackage{caption}
\usepackage{subcaption}
\usepackage{ragged2e}
\usepackage{hyperref} 
\usepackage{placeins}
\usepackage[utf8]{inputenc} 
\usepackage[T1]{fontenc}    
\usepackage{hyperref}       
\usepackage{url}            
\usepackage{booktabs}       
\usepackage{amsfonts}       
\usepackage{nicefrac}       
\usepackage{microtype}      
\usepackage{amsmath}

\usepackage{mathrsfs}  
\usepackage{bm}  
\usepackage{comment}
\usepackage{mathtools}
\usepackage{cleveref}
\usepackage{placeins}
\usepackage{amsopn}

\usepackage{algorithm}
\usepackage{algpseudocode}

\newtheorem*{informaltheorem}{Informal Theorem}

\title[Pseudo-Markovian Alternative to Mori-Zwanzig]{Surrogate Trajectories Along Probability Flows: \\ Pseudo-Markovian Alternative to Mori–Zwanzig}
\author{Noé Stauffer, Hossein Gorji and Ivan Lunati}
\thanks{
Emails: \href{mailto:noe.stauffer@empa.ch}{\texttt{noe.stauffer@empa.ch}}, 
\href{mailto:mohammadhossein.gorji@empa.ch}{\texttt{mohammadhossein.gorji@empa.ch}}, 
\href{mailto:ivan.lunati@empa.ch}{\texttt{ivan.lunati@empa.ch}}.
Noé Stauffer: Laboratory for Computational Engineering, Empa, Dübendorf, Switzerland and CSQI Chair, École Polytechnique Fédérale de Lausanne, Lausanne, Switzerland. Hossein Gorji: Laboratory for Computational Engineering, Empa, Dübendorf, Switzerland. Ivan Lunati: Laboratory for Computational Engineering, Empa, Dübendorf, Switzerland. }

\begin{document}
\maketitle

\begin{abstract}
Model reduction techniques have emerged as a powerful paradigm across different fronts of scientific computing. Despite their success, the provided tools and methodologies remain limited if high-dimensional dynamical systems subject to initial uncertainty and/or stochastic noise are encountered; in particular if rare events are of interest. We address this open challenge by borrowing ideas from Mori-Zwanzig formalism and Chorin's optimal prediction method. The novelty of our work lies on employing time-dependent optimal projection of the dynamic on a desired set of resolved variables. We show several theoretical and numerical properties of our model reduction approach. In particular, we show that the devised surrogate trajectories are consistent with the probability flow of the full-order system. Furthermore, we identify the measure underlying the projection through polynomial chaos expansion technique. This allows us to efficiently compute the projection even for trajectories that are initiated on low probability events. Moreover, we investigate the introduced model-reduction error of the surrogate trajectories on a standard setup, characterizing the convergence behaviour of the scheme. Several numerical results highlight the computational advantages of the proposed scheme in comparison to Monte-Carlo and optimal prediction method. Through this framework, we demonstrate that by tracking the measure along with the consistent projection of the dynamic we are able to access accurate estimates of different statistics  including observables conditional on a given initial configuration.
\end{abstract}
\section{Introduction}
\subsection{Background}
\noindent Model reduction techniques aim to reduce the computational burden associated with solving dynamical systems with many degrees of freedom. 
They reduce the computational complexity of full-order numerical solutions by approximating them in a lower dimensional subspace while retaining adequate accuracy. These methods differ in how the reduced space and its basis are constructed, and may be equipped with corresponding error estimates. Many of these techniques were built upon and extend the foundational works on Proper Orthogonal Decompositions (POD) \cite{POD, rathinam,rowley}, Krylov-based methods \cite{BAI20029,SCHMID_2010} and (certified) Reduced Basis Methods (RBM) \cite{quarteroni,hesthaven2015model}. \\ \ \\
These approaches rely predominantly on linear projection operators, while extensions to nonlinear projections have beome an active area of research, driven by advances in data-driven and machine-learning methodologies \cite{qian2019transform,qian2022reduced,willcox2024distributed,lee2020model}, with notable recent progresses reported in \cite{schwerdtner2024greedy,otto2024role,buchfink2024model}. Most classical frameworks focus on systems of Ordinary Differential Equations (ODEs) resulting from the discretization of Partial Differential Equations (PDEs), particularly when solutions are sought for varying parameters and/or initial conditions rather than for fixed values.  \\ \ \\ 
Despite the significant progress, conventional model reduction techniques are not, in their original form, well suited for stochastic settings, as parameters and initial conditions are typically not treated as random variables. The body of work on model reduction for Stochastic Differential Equations (SDEs) is even more limited than for ODEs, due to the infinite-dimensional stochasticity of Brownian motion. \\ \ \\ 
In general, three main strategies exist to compute the probability flow of stochastic systems, described either by the Liouville equation \cite{Evans_Morriss_2008} for ODEs with uncertain initial condition, or by the Fokker-Planck (FP) equation \cite{risken1989fokker} for SDEs. 
The first strategy employs Partial Differential Equation (PDE) solvers to obtain full-order numerical solution to these equations, but it is limited to low-dimensional systems due to the curse-of-dimensionality. 
The second relies on Monte-Carlo (MC) sampling \cite{kloeden2011stochastic,rogers1996diffusions} to alleviate the curse-of-dimensionality; while convergence is often dimension-independent, MC methods suffer from prohibitively slow convergence when rare events are of interest \cite{budhiraja2019analysis}. The third strategy uses reduction methods originating from statistical physics \cite{gallavotti2013statistical,ellis2012entropy}, which adopt a statistical perspective by considering ensembles of micro-states. In projection operator techniques \cite{grabert1982proj,_ttinger_2005,Vrugt_2020}, the infinite-dimensional dynamics of the Liouville equation is projected onto a Hilbert space spanned by a set of relevant observables, for which moment transport equations and statistical properties are derived under certain closure assumptions. Even with a suitable closure, the resulting moment transport equations cannot provide statistics conditional on specific initial configurations. In other words, they do not yield trajectories for the observables.  \\ \ \\
In a series of studies \cite{chorin2000MZ,CHORIN2009optimal_prediction_memory}, Chorin and co-workers  bridged statistical physics and model reduction by devising a framework for ODE systems with random initial conditions, offering a path to construct trajectories for the observables. The approach is built on the Mori-Zwanzig (MZ) formalism \cite{Evans_Morriss_2008,gouasmi2017memory,Givon2005orthogdyn} originating from statistical physics \cite{Evans_Morriss_2008}.   
In this framework, the reduced space is specified by the choice of variables of interest, referred to as "resolved variables". The evolution of these variables is described by projecting the full dynamics onto the reduced space they span. The resulting reduced model is formally expressed, though not in closed form, by the generalized Langevin equation (GLE). 
Defining the conditional expectation as the projection operator yields first-order optimal prediction, which can be further extended to incorporate the memory kernel \cite{chorin2000MZ,CHORIN2009optimal_prediction_memory}.  
Under specific assumptions on the unclosed terms, this framework allows access to the statistics of trajectories for the resolved variables, conditional on their initial conditions.
Building on this approach and methods for its practical implementation, various applications have been explored \cite{WANG2020109402,LIN2021109864,jeong2024diffgledifferentiablecoarsegraineddynamics,Parish}.
However, the bottleneck of the formalism lies in the computation of the memory integral, whose complexity is similar to that of the original high-dimensional dynamics and is generally intractable \cite{gouasmi2017memory}, often necessitating crude approximations.  \\ \ \\ 
The absence of trajectories in moment based model reductions, the inefficiency of Monte-Carlo methods for rare events, and the computational intractability of the MZ memory term constitute central shortcomings of current reduction techniques. To address these limitations, we propose a reduction framework, applicable to both ODEs and SDEs, based on a time-dependent projection operator defined by the conditional expectation with respect to the current configuration of the resolved variables. This choice eliminates the memory integral while ensuring that trajectories of the reduced system remain consistent with the marginal probability flow of the full dynamics \cite{gyongy1986mimicking, brunick2013mimicking}. Our method combines the interpretability of projection operator techniques with the practicality of trajectory-based approaches, providing a concrete workflow for applying MZ-inspired reductions to high-dimensional stochastic systems.
\subsection{Our contribution and organization}
\noindent This paper directly addresses the following desiderata, which current model reduction techniques fail to satisfy simultaneously.
\begin{enumerate}
\item {\bf Trajectory construction.} Most projection-based reductions of stochastic systems provide only moment equations and lack trajectories for observables, yet constructing trajectories is essential for accessing statistics conditional on initial states.
\item {\bf Computational tractability.} Although MZ formalism offers a principled reduction framework, its memory integral is generally intractable, preventing practical applications to high-dimensional systems.
\item {\bf Rare-event flexibility.} Convergence of Monte-Carlo estimators is prohibitively slow when conditioning on low-probability initial states, limiting their usefulness for rare-event analysis.
\item {\bf Extension to SDEs.}  Many reduction techniques are designed for deterministic ODEs and their extension to SDEs with drift–diffusion dynamics is not straightforward.
\end{enumerate}
We overcome these limitations by introducing a model reduction framework applicable to both ODEs and SDEs, built upon a time-dependent projection operator defined by conditional expectations with respect to the current probability measure. Our contributions can be summarized as follows.
\begin{enumerate}
\item {\bf Surrogate trajectories.} We construct trajectories for resolved variables that are exact in law, ensuring consistency with the marginal probability flow of the full dynamics. 
\item {\bf Optimal projection without memory.} Our conditional expectation operator is optimal in $L^2$ and eliminates the memory term of the MZ formalism, thereby rendering the reduction computationally tractable.
\item {\bf Accessible rare events.} By combining Polynomial Chaos (PC) expansions with conditional expectations, our framework efficiently computes reduced dynamics even when the resolved variables lie in regions of vanishing probability.
\item {\bf General applicability.} Our framework extends from deterministic ODEs with uncertain initial data to SDEs and, from subsets of state variables to general observables. Once the probability flow of the full system is estimated, surrogate trajectories for arbitrary resolved variables can be constructed at negligible online cost.
\end{enumerate}
The remainder of the paper is organized  as follows.
\begin{enumerate}
\item Section \ref{sec:notation and main problem} introduces the problem setting, notation, and formulates the main research question.
\item Section \ref{sec:main idea} presents the central idea underlying our reduction technique. 
\item Section \ref{sec:optimal proba consitstent reduced dynamics} develops and justifies the theoretical results and their connection to the MZ formalism.
\item  Section \ref{sec:numerics} presents the numerical framework, including the use of PC expansions to solve the probability flow in the offline phase and the computation of conditional expectations in the online phase.
\item Section \ref{sec:numerical experiments} reports numerical experiments validating our approach, empirically demonstrating  the convergence behaviour, and illustrating efficiency across different regimes, including rare events. 
\item Section \ref{sec:conclusions}, summarizes the main findings, discusses the limitations, and outlines possible directions for future work. 
\end{enumerate}
\section{Preliminaries}\label{sec:notation and main problem}
\subsection{Notation}
\noindent Throughout the paper, we adopt the following conventions. Coordinates in the state space are denoted by lowercase letters, e.g., $x$, while random variables are written in uppercase, e.g., $X$, unless stated otherwise. A random variable $X$ with law (probability distribution) $\mu$ is an element of $L^2(\Omega,\mathcal{A},\mathbb{P})$, the Hilbert space of square-integrable random variables on the probability space $(\Omega,\mathcal{A},\mathbb{P})$. A stochastic process is a family $\{X_t\}_{t\in\mathcal{T}}$ of random variables indexed by time, whose marginal laws $\{\mu_t\}_{t\in\mathcal{T}}$ define the probability flow. When $\mu_t$ is absolutely continuous with respect to the Lebesgue measure, we denote its probability density by $\rho_t$, so that $d\mu_t(x)=\rho_t(x)\,dx$. We denote by $L^2(\mathbb{R}^d,\mu)$ the space of measurable functions on $\mathbb{R}^d$ that are square-integrable with respect to $\mu$. \\ \ \\
As polynomial bases will be employed later to represent random variables (e.g., through PC expansions) we introduce the notation of the multi-index  $a=(a_1,\ldots,a_n)\in\mathbb{N}^n$, which has length $|a|=\sum_{i=1}^n a_i$ and defines monomials $x^a=\prod_{i=1}^n x_i^{a_i}$. For an integer $r\ge0$, we denote by $\mathcal{I}_r^n=\{a\in\mathbb{N}^n:|a|\le r\}$ the set of all multi-indices of degree at most $r$, so that a polynomial of degree $\le r$ can be written as $\sum_{a\in \mathcal{I}_r^n} c_a x^a$ with coefficients $c_a\in\mathbb{R}$. \\ \ \\
Given a random variable $X$ and a sub-$\sigma$-algebra $\mathcal{G}\subseteq\mathcal{A}$, we denote by $\mathbb{E}[X|\mathcal{G}]$ the conditional expectation, which represents the orthogonal projection of $X$ onto the subspace of $\mathcal{G}$-measurable random variables. When conditioning on another random variable $Y$, we write $\mathbb{E}[X|Y]$ as shorthand for $\mathbb{E}[X|\sigma(Y)]$, where $\sigma(Y)$ is the smallest $\sigma$-algebra with respect to which $Y$ is measurable. For a state vector $X\in\Gamma\subseteq\mathbb{R}^N$ with law $\mu$, and a measurable map $A:\Gamma\to\mathbb{R}^m$ defining the resolved variables 
$\hat X=A(X)$,
the law of $\hat X$ is the pushforward $\hat{\mu}=A_\#\mu$. The conditional expectation $\mathbb{E}[g(X)|\hat X]$ is then the orthogonal projection of $g\in L^2(\Gamma,\mu)$ onto $L^2(\mathbb{R}^m,\hat\mu)$. We denote this projection by $P[g](\hat X)$ and its complement by $Q[g](\hat X)$. We further introduce a {time-dependent projection operator} $P_t$, defined as the conditional expectation with respect to the current probability law $\mu_t$, i.e. 
\begin{eqnarray}
P_t g &=& \mathbb{E}_{\mu_t}[\,g(X_t)\mid \hat{X}_t\,]. \nonumber
\end{eqnarray}

\noindent For dynamical systems, we consider both deterministic and stochastic cases.
Let $(\Omega, \mathcal{F},(\mathcal{F}_t)_{t\geq 0},P)$ be a filtered probability space, $\Gamma \subset\mathbb{R}^N$ the phase space and the state vector $X\in\Gamma\subseteq\mathbb{R}^N$.
Consider the deterministic dynamics,
\begin{eqnarray}\label{eq:ode assumption}
    dX_t = b(X_t,t)dt, \qquad    X_0\sim \mu_0,
\end{eqnarray}
with flow map $X_t=\Phi(X_0,t)$ and generator (Liouville operator) $L=\sum_{i=1}^N b_i \,\partial_{x_i}$. For such systems, we denote by $e^{tL}$ the semigroup of operators generated by $L$, so that $u(t)=e^{tL}u(0)$ solves the corresponding evolution equation.
A stochastic system is modeled by the $\Gamma$-valued stochastic process $\{X_t\}_{t\geq 0}, \ X_t\in L^2(\Omega;\mathbb{R}^N)$, governed by the Itô SDE,
\begin{eqnarray}\label{eq:sde assumption}
    dX_t = b(X_t,t)dt + \sigma(X_t,t)dW_t,\qquad    X_0\sim \mu_0,
\end{eqnarray}
with the drift vector $b:\Gamma\times\mathbb{R}\rightarrow\mathbb{R}^N$, and where the noise term arises from the $q$-dimensional Wiener process $W_t$ with independent components and the $N\times q$ matrix of noise coefficients $\sigma:\Gamma\times\mathbb{R}\rightarrow \mathbb{R}^{N\times q}$. In this stochastic case, the associated generator is
\begin{eqnarray}
    L=\sum_i b_i \,\partial_{x_i} + \sum_{i,j} D_{ij}\,\partial_{x_i}\partial_{x_j}, \qquad D=\tfrac12 \sigma\sigma^\top .
\end{eqnarray}
We also consider the corresponding probability flows,
\begin{equation}\label{eq:mut}
    \partial_t\mu_t=\mathcal{L}\mu_t,
\end{equation}
where $\mathcal{L}=L^*$ consists of the Kolmogorov forward operator $\mathcal{L}=-\sum_{i=1}^N\partial_{x_i}b_i$ or $\mathcal{L}=-\sum_{i=1}^N\partial_{x_i}b_i+\sum_{i=1}^N\partial_{x_i}\partial_{x_j}D_{ij}$, respectively, and $L^*$ denotes the adjoint of the generator $L$ of either process.
We define the resolved variables $\hat{X}:=A(X)=\{A_i(X)\}_{i=\{1,\dots,m\}}\in\hat{\Gamma}, \ A_i:\Gamma\to\hat{\Gamma}_i\subseteq\mathbb{R}$, where $\hat{\Gamma}=\hat{\Gamma}_1\times\dots\times\hat{\Gamma}_m$.
\subsection{Main problem}
\noindent Consider the high-dimensional deterministic system with uncertain initial condition
\begin{eqnarray} \label{eq:ode1}
\dot{X}_t &=& b(X_t,t), \ \  \qquad X(0)=X_0 \sim \mu_0, \quad X\in\Gamma\subseteq\mathbb{R}^N.
\end{eqnarray}
We are interested in the evolution of a set of resolved variables $\,\hat{X}\in\hat{\Gamma}\subseteq\mathbb{R}^m$ ($m\ll N$), with the full state decomposed as $X=(\hat{X},\hat{X}^{\perp})$. The central question is: { How can one construct surrogate trajectories for $\mathbb{E}[\hat{X}_t|\hat{X}_0]$ that are exact in law, i.e. consistent with the marginal probability flow $\hat{\mu}_t$ of the resolved variables?}
Classical approaches face following limitations.
\begin{itemize}
\item The MZ formalism provides an exact reduction, but introduces an intractable memory integral that requires knowledge of the full system.
\item Monte Carlo sampling is exact in law but converges slowly, and becomes impractical when conditioning on rare events. 
\item Solving the Liouville or Fokker–Planck equations yields the probability flow, but does not provide trajectories or multi-time statistics.
\end{itemize}
As a motivating example, consider Hamiltonian molecular dynamics: although the full system may involve thousands of degrees of freedom, one is often interested only in a few relevant or collective coordinates. Similarly, in network dynamics, attention may focus on the evolution of a handful of hub nodes rather than describing all peripheral components. In both cases, reduced models that remain consistent in law with the marginal statistics of the selected variables are crucial for efficient prediction and analysis. \\ \ \\
Our goal is to design a framework that overcomes these obstacles and extends naturally from deterministic ODEs to stochastic drift–diffusion SDEs
\begin{eqnarray}\label{eq:sde1}
dX_t &=& b(X_t,t)\,dt + \sigma(X_t,t)\,dW_t, \qquad X_0 \sim \mu_0,
\end{eqnarray}
while preserving consistency in law for the surrogate dynamics of $\hat{X}_t$.
\section{Main idea}\label{sec:main idea}
\noindent  
To construct the reduced dynamics, the key ingredient is to use the {time-dependent projection operator}
\begin{eqnarray}
\label{eq:pro-time-dependent}
(P_t g)(\hat{X}_t) &=& \mathbb{E}_{\mu_t}[\,g(X)\mid \hat{X}=\hat{X}_t\,], \qquad g\in L^2(\Gamma,\mu_t),
\end{eqnarray}
defined as the conditional expectation with respect to the current probability measure $\mu_t$ and the current state of $\hat{X}_t$. This operator provides an optimal $L^2$-projection at each time and evolves adaptively along the probability flow, as illustrated in Figure \ref{fig:idea}. \\ \ \\
Applying $P_t$ to the drift of the resolved variables, $\hat{b}(X,t)$, yields the reduced dynamics
\begin{eqnarray}
\frac{d}{dt}\hat{X}_t&=& \mathbb{E}_{\mu_t}[\,\hat{b}(X,t)\mid \hat{X}=\hat{X}_t\,], \qquad \hat{X}_0\sim \hat{\mu}_0.
\end{eqnarray}
The resulting surrogate trajectories evolve consistently with the marginal law $\hat{\mu}_t$ of $\hat{X}_t$, and are therefore exact in distribution.  This property is established formally in Section \ref{sec:optimal proba consitstent reduced dynamics}. \\ \ \\
In contrast to the MZ formalism, our approach eliminates the intractable memory integral by continuously updating both the probability measure and the reference configuration, thereby ensuring that the random force remains orthogonal to the resolved space by construction and that the reduced dynamics are closed without additional assumptions (see Section \ref{sec:MZ}).
\begin{figure}
\includegraphics[width=\columnwidth]{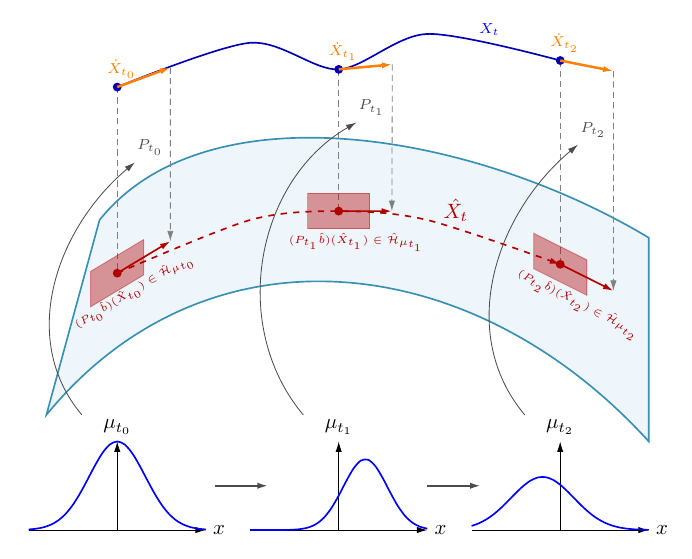}%
\caption{\label{fig:idea} \justifying Surrogate trajectory of the resolved variables $\hat{X}$ along the probability flow of the full-order probability measure $\mu_t$.}%
\end{figure}
\\ \ \\ In practice, the probability law $\mu_t$ is not known explicitly, which raises the central challenge of computing the conditional expectations that define the reduced dynamics. Intuitively, it is often easier to approximate the evolution of the probability measure than that of individual trajectories: while sample paths may fluctuate irregularly, the associated measures typically evolve smoothly in time. To implement the framework, we approximate the probability flow $\mu_t$ using a Polynomial Chaos (PC) expansion of the state
\begin{eqnarray}
X_t &\approx& \sum_{\alpha} c_{\alpha,t} H_\alpha(\xi),
\end{eqnarray}
where, e.g. $\xi$ is a standard Gaussian random variable and $\{H_\alpha\}$ Hermite polynomials. The PC coefficients $c_{\alpha,t}$ evolve deterministically, resulting in an efficient representation of $\mu_t$. Conditional expectations are then computed by polynomial regression on the resolved variables, avoiding costly Monte Carlo sampling in low-probability regions. \\ \ \\
In summary, our method combines the following.
\begin{enumerate}
\item  A time-dependent optimal projection ensuring consistency in law.
\item Efficient computation of the probability flow via PC expansion. 
\item Affordable estimation of conditional expectation through polynomial regression. 
\end{enumerate}
The resulting workflow constitutes a trajectory-based reduction framework that retains the interpretability of projection methods while remaining  tractable for both ODEs and SDEs. Before establishing the theoretical results of this framework in Section \ref{sec:optimal proba consitstent reduced dynamics}, we present a detailed comparison with first-order optimal prediction and the MZ formalism.
\section{Motivation and connection to the Mori–Zwanzig formalism}\label{sec:MZ}
\subsection{Pseudo-Markovian decomposition}
\noindent The starting point of our framework is the following key observation (see Theorem~\ref{thm:pseudo-mark} for its formal presentation). 
\begin{informaltheorem}[Pseudo-Markovian decomposition]
Consider a high-dimensional system $\dot{X}_t=b(X_t)$ with probability law $\mu_t$, and decompose the state  into resolved, $\hat{X}$, and unresolved variables, $\hat{X}^{\perp}$.  
If the drift $b$ is projected onto functions of the resolved variables using the conditional expectation with respect to the current measure $\mu_t$, then the resolved dynamics can be written as
\[
\dot{\hat{X}}_t \;=\; \mathbb{E}_{\mu_t}[\,{b}(X)\mid \hat{X}_t\,] \;+\; \zeta_t,
\]
where $\zeta_t$ is a noise term orthogonal to the resolved space.  
\end{informaltheorem}
\noindent In other words, the reduced dynamics decomposes into a projected drift plus an orthogonal fluctuation. This decomposition is {pseudo-Markovian}: it evolves forward in time using only the current state of the resolved variables, yet it remains exact, because the time-dependent noise captures the effect of unresolved degrees of freedom.   The formal version, with assumptions and convergence guarantees, is presented in Section \ref{sec:optimal proba consitstent reduced dynamics}. 
\subsection{Link to the Mori–Zwanzig formalism}
\noindent The above decomposition is closely related to the MZ framework. In classical MZ \cite{chorin2000MZ,CHORIN2009optimal_prediction_memory,gouasmi2017memory,Givon2005orthogdyn,LIN2021109864}, the evolution of the resolved variables $\hat{X}$ is expressed through a generalized Langevin equation (GLE)
\begin{equation}\label{eq:MZ}
\frac{d}{dt} e^{tL}\hat{X} = e^{tL}Pb(\hat{X}) 
+ \int_0^t e^{(t-s)L}PL e^{sQL} Q b(\hat{X})\,ds
+ e^{tQL}Q b(\hat{X}), \qquad \textrm{ with   } Q=1-P,
\end{equation}
which comprises a Markovian term, a memory integral, and a random force orthogonal to the resolved space. While exact, the memory integral requires solving orthogonal dynamics and is intractable in practice. 
\subsection{First-order optimal prediction}\label{sec:FOP}  
\noindent Chorin and co-workers \cite{chorin2000MZ} proposed a practical closure by projecting with respect to the \emph{initial} distribution $\mu_0$, i.e. $P g=\mathbb{E}_{\mu_0}[g\mid \hat{X}_0]$, and neglecting the memory integral. This yields explicit ODEs for $\mathbb{E}[\hat{X}_t\mid \hat{X}_0]$, often called the first-order optimal prediction. While effective over short times, this approximation relies on the assumption that conditional expectations at time $t$ can be replaced by those at $t=0$, which limits its accuracy  at later times. 
\subsection{Our construct}  
\noindent Time-dependent projections have also been considered by Grabert \cite{grabert1982proj} and other authors in the physics literature \cite{_ttinger_2005,Vrugt_2020}, but in those settings the projection remains tied to the initial configuration, causing the orthogonality of the noise to be lost.  \\ \ \\ 
In contrast, our conditional-expectation projection $P_t$ is defined with respect to the {current law} $\mu_t$ and the {current resolved configuration} $\hat{X}_t$. This ensures both a vanishing memory integral  and an orthogonal noise term at all times, by virtue of the optimality of conditional expectation.  The resulting reduced dynamics,
\begin{eqnarray}\label{eq:our construct}
\dot{\hat{X}}_t = \mathbb{E}_{\mu_t}[\,{b}(X)\mid \hat{X}_t\,] +\zeta_t, 
\qquad \mathbb{E}_{\mu_t}[\,\zeta_t\mid \hat{X}_t\,]=0,
\end{eqnarray}
is therefore memory-free, exact in distribution, and consistent with the marginal probability law of the resolved variables. 
\\ \ \\
 This provides the conceptual foundation for our probability-consistent reduced dynamics, which can be viewed as a probability-adapted, pseudo-Markovian realization of the MZ formalism, in which time dependence replaces the memory integral. 
 Intuitively, our method acts as a continuous restart of the MZ formalism, constantly updating the initial reference configuration and leaving the memory integral negligible (growing as $O(\Delta t^2)$).
 Since the noise term remains unclosed (in both MZ and Pseudo-Markovian decomposition) and requires separate modeling assumption, we adopt the simple choice $\zeta_t=0$.  As shown in Section \ref{sec:optimal proba consitstent reduced dynamics}, although this choice yields approximate trajectories, it still preserves consistency in marginal law, as first emphasized by Gyöngy \cite{gyongy1986mimicking} for a one-dimensional Itô process, and further generalized by Brunick and Shreve \cite{brunick2013mimicking}. In other words, the dynamics of the resolved variables are described by surrogate trajectories $\hat{X}_t$ that remain exact in distribution with respect to the marginal law $\hat{\mu}_t$.
\section{Theoretical results}\label{sec:optimal proba consitstent reduced dynamics}
\subsection{Assumptions}\label{sec:assumptions}
\noindent The theoretical results of this section are subject to the following list of assumptions on the regularity of the drift and noise coefficients as well as on the probability flow $\{\mu_t\}_{0\leq t\leq T}$ Eq.~(\ref{eq:mut}) associated to Eqs.~(\ref{eq:ode assumption})-(\ref{eq:sde assumption}). 
For convenience we consider autonomous dynamical systems but the results can be generalized to non-autonomous dynamical systems provided sufficient regularity and controlled growth. 
For each theoretical result we mention which assumptions from the list below are needed.
\begin{enumerate}
    \item The vector field $b:\Gamma \to\mathbb{R}^N$ is globally Lipschitz on $\Gamma$, and continuous in time. 
    \item The noise coefficient $\sigma:\Gamma\to\mathbb{R}^{N\times N}$ is globally Lipschitz on $\Gamma$, and continuous in time. 
    \item The initial probability measure has finite second moment and absolutely continuous with respect to the Lebesgue measure i.e. $\mu_0\in\mathcal{P}^r_2(\Gamma)$.
    \item The conditional measures $\{\mu_{t|\hat{X}_t=x}\}$ exist and $\forall t\geq 0$ there exists $K_t\geq 0$ such that,
    \begin{eqnarray}
    \mathcal{W}_1(\mu_{t|\hat{X}_t=y},\mu_{0|\hat{X}_t=z})\leq K_t|y-z|, \quad \forall \ y,z \in \hat{\Gamma},
    \end{eqnarray}
    where $\mathcal{W}_1$ is the 1-Wassertein distance. 
\end{enumerate}
Note that the constraint on stability of the conditional measure results in the following. 
\begin{proposition}[Lipschitzness of Conditional Expectation]\label{prop:lipschitz mckean-vlasov}
Given the assumption (4), for any globally Lipschitz function $f:\Gamma\to\mathbb{R}$, there exists $C_t\geq 0, \ \forall t\geq0$ such that,  
\begin{eqnarray}
        \big|\mathbb{E}_{\mu_t}[f(X)\mid \hat{X}=y] - \mathbb{E}_{\mu_t}[f(X)\mid \hat{X}=z]\big|\leq C_t|y-z|,\: \forall \ y,z\in\hat{\Gamma}.
\end{eqnarray}
\end{proposition}
\begin{proof}
    See Appendix~\ref{app:condexp}
\end{proof}
\subsection{Pseudo-Markovian dynamics}
\subsubsection{Exact trajectory}
\begin{theorem}[Trajectory Consistency for ODE]\label{thm:pseudo-mark}
    Consider the dynamical system Eq.~(\ref{eq:ode assumption}) and the choice of resolved variables $\hat{X}:=X_i$ for some $i\in \{1,...,N\}$. Let $\hat{b}:=b_i$ and $\hat{X}^0=X_i(t_0)$. Given assumptions (1), (3), we have the following.
    \begin{enumerate}
        \item Under suitable choice of $\zeta^n$ and for small $\Delta t>0$, the numerical scheme 
        \begin{eqnarray}\label{eq:scheme proposition}
            \hat{X}^{n+1} = \hat{X}^{n} +  \mathbb{E}_{\mu_t}[\hat{b}(X)\mid\hat{X}^n]\Delta t  + \zeta^n\Delta t, 
        \end{eqnarray}
        converges in the following sense: $\mathbb{E}[|\hat{X}^{n}-X_i(t_n)|^2]=O(\Delta t^2)$, where $X_i(t_n)$ is the exact solution at time $t_n$.
        \item The term $\zeta^n$ can be identified as a noise term in the following sense,
        \begin{eqnarray}
        \label{eq:noise-iden}
            \mathbb{E}_{\mu_t}[\zeta^n] = 0, \ 
            \mathbb{E}_{\mu_t}[\zeta^n\mid\hat{X}^n] = 0, \ \ \ \textrm{and} \ \ 
            \mathbb{E}_{\mu_t}\left[\zeta^n\:\mathbb{E}_{\mu_t}[\hat{b}(X)\mid\hat{X}^n]\right] = 0,
        \end{eqnarray}
        $\forall \:t\geq 0$.
    \end{enumerate}
\end{theorem}
\begin{proof}
    See Appendix.~\ref{proof:pseudomark}.
\end{proof}
\begin{corollary}[Trajectory Consistency for SDE] \label{corol:pseudo-mark-sde} Similar result would hold if $X_t$ is solution of the SDE Eq.~(\ref{eq:sde assumption}) fulfilling assumptions (1)-(3) .
\begin{enumerate}
\item For suitable choices of noise terms $\zeta^n$ and $\eta^n$, the scheme 
 \begin{eqnarray}
            \hat{X}^{n+1} = \hat{X}^{n} +  \mathbb{E}_{\mu_t}[\hat{b}(X)\mid\hat{X}^n]\Delta t+\mathbb{E}_{\mu_t}[\hat{\sigma}(X)|\hat{X}^n] \sqrt{\Delta t}\xi^n+\left(\zeta^n\Delta t+\eta^n \sqrt{\Delta t}\right), 
        \end{eqnarray}
        with $\xi^n$ as the standard normal random number, converges in the following sense: given identical Brownian path for $\hat{X}$ and $X$, $\mathbb{E}[|\hat{X}^n-X_i(t_n)|^2]=O(\sqrt{\Delta t})$. 
\item The terms $\zeta^n$ and $\eta^n$ can be identified as noise terms in the sense of fulfilling Eq.~\eqref{eq:noise-iden}.
\end{enumerate}        
\end{corollary}
\begin{proof}
    See Appendix.~\ref{proof:pseudomark sde}.
\end{proof}
\noindent \textbf{Remark:} these results can be extended to arbitrary $\mathcal{F}$-measurable smooth observables $A=(A_i)_{1\leq i\leq m}$, $A:\Gamma\to\hat{\Gamma}\in\mathbb{R}^m$ provided $LA$ and $\sum_{i,j=1}^ND_{ij}\frac{\partial A_k \partial A_l}{\partial x_j \partial x_i}$ are Lipschitz.

\subsubsection{Surrogate trajectory}
\begin{proposition}[Marginal Preservation: Deterministic Case]\label{prop:deterministic}
Let $X\in\Gamma\subseteq\mathbb{R}^N$ evolve according to the deterministic dynamics, Eq.~(\ref{eq:ode assumption}), with uncertain initial condition. Assume $b\in \mathcal{H}_t=L^2(\Gamma,\mu_t), \forall t \geq 0$ and let $\hat{b}:=(b_1,\dots,b_m)$. Let surrogate trajectories for the resolved variables $\hat{X}=(X_1,\dots,X_m)\in\mathbb{R}^m, \ m<N$ satisfy the reduced dynamics,
\begin{eqnarray}\label{eq:surrogate dynamics proposition}
    d\hat{X}_t = \mathbb{E}_{\mu_t}[\hat{b}(X)\mid\hat{X} = \hat{X}_t] \ dt, \ \qquad
    \hat{X}_0 \sim \hat{\mu}_0(d\hat{X}),
\end{eqnarray}
where we defined the initial marginal law $\hat{\mu}_0(\hat{X})=\int_{\hat{\Gamma}^\perp}d\mu_0(X)$, $\hat{X}\in\hat{\Gamma}$ and $\Gamma=\hat{\Gamma}\times\hat{\Gamma}^\perp$.
Then, given assumptions (1), (3)-(4), the surrogate trajectories preserve the marginal law $\hat{X}_t\sim\hat{\mu}_t, \forall t\geq 0$. 
Moreover, at each instant $t$, $\mathbb{E}_{\mu_t}[\hat{b}(X)\mid\hat{X} = \hat{X}_t]$ is the optimal projection of the vector field $b$ in $\mathcal{H}_t=L^2(\mathbb{R}^N,\mu_t)$, where $\mu_t$ is the probability measure of the full-order system at time $t$. 
\end{proposition}
\begin{proof}
    See Appendix.~\ref{app:proof x ode}.
\end{proof}
\begin{proposition}[Marginal Preservation: Stochastic Case]\label{prop:stochastic}
Consider the $\Gamma$-valued stochastic process $\{X_t\}_{t\geq 0}, \ X_t\in L^2(\Omega;\mathbb{R}^N)$, governed by the Itô SDE Eq.~(\ref{eq:sde assumption}). 
Suppose that $b,\sigma \in\mathcal{H}_t=L^2(\mathbb{R}^N,\mu_t),\forall t \geq 0$. Define the resolved variables $\hat{X}:=(X_1,\dots,X_m)\in\mathbb{R}^m, \ m<N$ and let $\hat{b}:=(b_1,\dots,b_m)$.
Then, given assumptions (1)-(4), the surrogate trajectories satisfying the reduced dynamics,
\begin{eqnarray}\label{eq:optimal prob consistent sde}
    d\hat{X}_t = \mathbb{E}_{\mu_t}[\hat{b}(X)\mid\hat{X}=\hat{X}_t]dt + \hat{\sigma}(\hat{X}_t)d\hat{W}_t, \quad
    \hat{X}_0 \sim \hat{\mu}_0,
\end{eqnarray}
for the resolved variables $\hat{X}$ preserve the marginal law $\hat{X}_t\sim\hat{\mu}_t, \forall t\geq 0$.
The matrix $\hat{\sigma}\in\mathbb{R}^{m \times q}$ satisfies $\frac{1}{2}\hat{\sigma}\hat{\sigma}^T = \mathbb{E}_{\mu_t}[\hat{D}\mid\hat{X}]$, where $\hat{D}\in\mathbb{R}^{m\times m}, \: \hat{D}_{ij}=D_{ij}, \: i,j\in\{1,\dots,m\}$ and $D\in\mathbb{R}^{N\times N}$ is the diffusion matrix related to the probability flow of $\mu_t$.
$\hat{W_t}$ is a $q$-dimensional Wiener process with independent components.
Moreover, the drift coefficient $b$ and the diffusion matrix $D$ are optimally projected at each instant $t\geq 0$.
\end{proposition}

\begin{proof}
See Appendix \ref{app:surrogate sde}.
\end{proof}

\noindent \textbf{Remark:} these results can be extended to arbitrary $\mathcal{F}$-measurable smooth observables $A=(A_i)_{1\leq i\leq m}$, $A:\Gamma\to\hat{\Gamma}\in\mathbb{R}^m$ of finite second moment, i.e. $A_t\in\mathcal{H}_t$.
\section{Numerical scheme}
\label{sec:numerics}
\noindent We now describe the numerical procedure used to implement the proposed pseudo-Markovian reduction. We restrict the presentation to the case where the resolved variables  $\hat{X}$ are a subset of the state components; extension to more general observables is straightforward (see Section \ref{sec:optimal proba consitstent reduced dynamics}). The scheme consists of three main components: (i) the approximation of the evolving probability flow $\mu_t$, 
(ii) the evaluation of conditional expectations, and (iii) the time integration of the reduced dynamics. 
At each step, we emphasize which parts of the scheme rely on heuristic approximations and which are supported by well-established theoretical results. 
\subsection{Approximation of the probability flow}
\noindent The probability flow $\{\mu_t\}$ is in general not available in closed form. To approximate the evolving law $\mu_t$ of the state, we employ a PC expansion which targets the probability distribution rather than the individual state trajectories. 
Specifically, we represent a random variable $Y_t$ with distribution $\mu_t$ as a finite expansion,
\begin{equation}\label{eq:PC expansion}
Y_t \;\approx\; M_t(\xi) \;=\; \sum_{\alpha\in\mathcal{I}_p^d} c_{\alpha}(t) H_\alpha(\xi),
\end{equation}
where $\xi\in\mathbb{R}^d$ denotes a vector of i.i.d. standard Gaussian variables, $\{H_\alpha\}$ are multivariate Hermite polynomials orthonormal in $L^2(\mathbb{R}^d,\mathcal{N}(0,I))$, and $c_\alpha(t)\in\mathbb{R}^N$ are deterministic coefficient vectors.  This choice of Gaussian inputs and Hermite polynomials is made for convenience; all subsequent steps extend directly to other families of Askey polynomials orthogonal with respect to different reference measures.  
The evolving probability law $\mu_t$ is thus approximated by the law of the random map $M_t(\xi)$.  
\subsubsection{Evolution of the coefficients: ODE case}
\noindent If $Y_t$ satisfies the deterministic system
\[
\dot{Y}_t = b(Y_t,t), \qquad Y_0\sim \mu_0,
\]
then substituting the PC representation \eqref{eq:PC expansion} into the dynamics and applying Galerkin projection yields the ODE system that determines the PC coefficients
\begin{equation}\label{eq:coeff evolution ODE}
\dot{c}_{\alpha,i}(t) \;=\; \mathbb{E}\left[\, b_i(M_t(\xi),t) \, H_\alpha(\xi)\,\right],
\qquad i=1,\dots,N,
\end{equation}
where the expectation is taken with respect to the standard Gaussian measure on $\xi$.  
The initial coefficients $\{c_\alpha(0)\}$ are determined by the expansion of the initial condition $Y_0$.
\subsubsection{Evolution of the coefficients: SDE case}
\noindent For the case where $X_t$ solves the Itô SDE
\begin{equation}\label{eq:ito}
dX_t = b(X_t,t)\,dt + \sigma(X_t,t)\,dW_t, \qquad X_0 \sim \mu_0,
\end{equation}
the law $\mu_t$ can, under suitable assumptions \cite{gorjichaos}, be equivalently generated by the McKean–Vlasov ODE
\begin{equation}
dY_t = b(Y_t,t)\,dt - \tfrac{1}{2}\,\sigma^2(Y_t,t) 
\nabla_x \log\!\big(\sigma^2(Y_t,t)\,\rho_Y(y,t)\big)\big|_{y=Y_t}\,dt,
\end{equation}
where $\rho_Y$ is the density of measure induced by $Y_t$.
Inserting the PC expansion \eqref{eq:PC expansion} into this dynamics gives the evolution equations for the coefficients
\begin{align}\label{eq:coeff evolution SDE}
dc_{\alpha,i}(t) &= \mathbb{E}\left[\, b_i(M_t(\xi),t)\, H_\alpha(\xi)\,\right]\,dt \\
&\quad+ \tfrac{1}{2}\,\mathbb{E}\left[\, 
\sigma^2(M_t(\xi),t) 
\sum_{l=1}^N \frac{\partial M_{t,l}^{-1}}{\partial \xi_i} 
\frac{\partial H_\alpha(\xi)}{\partial \xi_l}
\,\right]\,dt, \nonumber
\end{align}
for $i=1,\dots,N$.  
Note that ${\partial M_{t,l}^{-1}}/{\partial \xi_i}$ denotes the $(l,i)$-entry of the inverse Jacobian of the random map $M_t$.  
\subsubsection{Initial conditions}
\noindent If the initial law $\mu_0$ is Gaussian, say $\mu_0=\mathcal{N}(m,\Sigma)$, the coefficients at $t=0$ can be initialized as
\[
c_0 = m, 
\qquad c_j = L_j,\; j=1,\dots,N,
\]
where $L\in\mathbb{R}^{N\times N}$ satisfies $\tfrac{1}{2}L^\top L = \Sigma$.  
For a non-Gaussian initial law $\mu_0$, the coefficients must be computed by projection of the initial random variable $X_0$ onto the PC basis for each multi-index $\alpha$, 
\[
c_\alpha(0) = \mathbb{E}_{\mu_0}\big[\, Y_0 H_\alpha(\xi)\,\big],
\]
where $\xi$ are the canonical random inputs associated with the chosen PC basis (e.g. Gaussian or uniform).  
In practice, the expectations ca be either (i) evaluated analytically, if the density of $Y_0$ is known and integrals are tractable, or (ii) approximated numerically via Monte Carlo or quadrature sampling from $\mu_0$.  
If $\mu_0$ is not compatible with the chosen basis (e.g. non-Gaussian $\mu_0$ with Hermite chaos), a measure transformation or a different PC family (e.g. Askey scheme polynomials orthogonal to $\mu_0$) may be preferable.  
\subsubsection{Remarks}
\begin{itemize}
\item \textbf{Well-established results:} PC expansions converge spectrally for sufficiently smooth solutions, and Gaussian initial conditions admit initial coefficients in closed form.  
\item \textbf{Heuristic choice:} The PC basis is truncated at degree $p$, balancing accuracy and computational cost.  
\item \textbf{Complexity:} A full PC expansion of the $N$-dimensional law requires $\binom{N+p}{p}$ coefficients.  
In our implementation, all coefficients are updated, leading to combinatorial scaling in $N$. 
However, the PC coefficients can be updated entirely in an {offline phase}.  
Once the probability flow $\{\mu_t\}$ is represented in terms of its PC coefficients, conditional expectations of arbitrary observables (and therefore surrogate trajectories for arbitrary choices of resolved variables) can be evaluated in an {online phase} without recomputing the probability flow.  
This separation is a major advantage over the MZ formalism, where the memory integral must be recomputed for each new set of resolved variables.   
Finally, since reduced dynamics only requires conditional expectations with respect to the resolved variables, it is in principle possible to design reduced PC expansions whose cost scales with $m$ (the resolved dimension) rather than $N$ (the full dimension). Developing such reduced-complexity schemes is left for future work. 
\end{itemize}
\subsection{Computation of conditional expectations}
\noindent The reduced dynamics requires evaluating conditional expectations of the form
\[
\mathbb{E}_{\mu_t}\big[\, b_j(X) \,\big|\, \hat{X}_t \,\big], 
\qquad j=1,\dots,m.
\]
Since these conditional expectations depend only on the resolved variables $\hat{X}_t$, we approximate them by projection onto a polynomial basis in $\hat{X}_t$.  
\subsubsection{Weak sense formulation}  
\noindent Instead of computing the conditional expectation pointwise, we approximate it in a \emph{weak sense}, i.e. for a chosen polynomial basis $\{\psi_\alpha(\hat{X})\}_\alpha$, $\mathbb{E}_{\mu_t}[\, b_j(X)\mid \hat{X}_t\,]\approx h_{\{\psi_\alpha\}}(\hat{X}_t)$ and we require
\begin{eqnarray}
\mathbb{E}_{\mu_t}\!\left[\, b_j(X)\,\psi_\beta(\hat{X}_t)\,\right] 
= \mathbb{E}_{\mu_t}\!\left[\, \mathbb{E}_{\mu_t}[\,b_j(X)\mid \hat{X}_t\,]\,\psi_\beta(\hat{X}_t)\,\right] \overset{!}{=} \mathbb{E}_{\mu_t}\left[\, h_{\{\psi_\alpha\}}(\hat{X}_t)\,\psi_\beta(\hat{X}_t)\,\right],
\end{eqnarray}
for all test functions $\psi_\beta$ in the basis. 
\subsubsection{Practical implementation.}  
\noindent We approximate $\mathbb{E}_{\mu_t}[\, b_j(X)\mid \hat{X}_t\,]$ by a polynomial expansion
\begin{eqnarray}
\label{eq:estimate_CE}
\mathbb{E}_{\mu_t}[\, b_j(X)\mid \hat{X}_t\,] &\approx& \sum_{\beta\in\mathcal{J}_r^m} a_{j,\beta}(t)\,\psi_\beta(\hat{X}_t),
\end{eqnarray}
where $\mathcal{J}_r^m$ is the set of multi-indices up to degree $r$.  
The coefficients $\{a_{j,\beta}(t)\}$ are determined by solving the linear system
\begin{eqnarray}\label{eq:linear system CE}
\sum_{\gamma\in\mathcal{J}_r^m} G_{\beta,\gamma}(t)\,a_{j,\gamma}(t)
= \mathbb{E}_{\mu_t}\!\left[\, b_j(X)\,\psi_\beta(\hat{X}_t)\,\right],
\qquad 
G_{\beta,\gamma}(t)=\mathbb{E}_{\mu_t}[\,\psi_\beta(\hat{X}_t)\psi_\gamma(\hat{X}_t)\,].
\end{eqnarray}
The Gram matrix $G(t)$ is symmetric and positive definite by construction, ensuring well-posedness.  
\subsubsection{Remarks}  
\begin{itemize}
\item  \textbf{Well-established result:} Since the conditional expectation is the $L^2$-(optimal) projection onto functions of $\hat{X}$, the polynomial expansion in $\hat{X}$ is equivalent to performing a PC expansion with respect to the marginal law $\hat{\mu}_t$. As the polynomial degree $r\to\infty$,
PC expansions converge in $L^2$ and therefore the approximation \eqref{eq:estimate_CE} converges in $L^2(\hat{\mu}_t)$ to the exact conditional expectation. 
\item  \textbf{Heuristic choice:} In practice, we truncate the basis at degree $r$, and expectations are approximated using quadrature rule.  
\item  \textbf{Complexity:} The methods scales with the number of basis functions $\binom{m+r}{r}$, i.e. polynomially in the degree $r$ and combinatorially in the resolved dimension $m$. Since dependence is only on $m$and not on the unresolved dimension $(N-m)$, the online trajectory integration remains computationally efficient. 
\end{itemize}
\subsection{Time integration}
\noindent We adopt an explicit Euler scheme for simplicity:
\begin{eqnarray}\label{eq:Reduced Dynamics Discretized}
\hat{X}_{t+\Delta t} = \hat{X}_t + \Delta t \, \mathbb{E}_{\mu_t}[\,\hat{b}(X)\mid \hat{X}_t\,].
\end{eqnarray}
Higher-order schemes (Runge–Kutta, implicit methods) can be employed if higher convergence rate accuracy is required. Here, we choose explicit Euler because the main source of numerical error arises from the truncation of the PC and the polynomial regression approximations rather than from the time discretization.  

\subsection{Complexity}
\noindent To clarify the computational structure of the proposed reduction method, we summarize in 
Table~\ref{tab:complexity} the dominant contributions to both the offline and online phases.  
The offline stage consists of estimating the time-dependent probability flow by evolving the 
PC coefficients of the full system. This step requires manipulating the 
entire high-dimensional PC basis and therefore inherits the combinatorial cost 
$\binom{N+P}{P}$ associated with total-order $P$ polynomials in $N$ variables. 
Although this scaling is expensive, it is incurred only once, and the resulting representation 
of the evolving probability law can be reused for any choice of resolved variables, observables, 
 and their initial conditions. \\ \ \\
Once the probability flow is available, the online phase becomes significantly cheaper. 
Evaluating the conditional expectations that define the time-dependent projection operator 
depends only on the resolved dimension $m$ and the degree $r$ of the regression basis.  
In the “fixed-mode’’ setting, where PC coefficients and projection bases are precomputed and 
compressed, the cost of propagating surrogate trajectories becomes essentially 
$O(N m)$, reflecting the fact that the high-dimensional structure of the distribution has 
already been encoded in the PC representation.  
\begin{table}[t]
\centering
\caption{Computational complexity of the proposed reduction method}
\vspace{0.15cm}
\begin{tabular}{llll}
\toprule
\textbf{Phase} & \textbf{Task} & \textbf{General Complexity} & \textbf{Fixed-Mode} \\
\midrule
\multirow{2}{*}{Offline} 
 & PC coefficient update 
 & $O\!\left( N^{2}\binom{N+P}{P}^{2} + \binom{N+3P}{3P} \right)$
 & $O(N^{2})$ \\[2mm]
 & Memory for PC expansion 
 & $O\!\left( \binom{N+P}{P}^{2} + \binom{N+3P}{3P} \right)$
 & $O(N)$ \\[1mm]
\midrule
\multirow{2}{*}{Online}
 & Conditional expectations 
 & $O\!\left( N \binom{m+r}{r}^{2}\binom{N+r}{r} 
   + m \binom{m+r}{r}\binom{N+r}{r} 
   + \binom{N+3P}{3P} \right)$
 & $O(N\,m)$ \\[2mm]
 & Memory for projection matrices 
 & $O\!\left( \binom{m+r}{r}\binom{N+r}{r}
   + m \binom{N+3P}{3P} \right)$
 & $O(m)$ \\
\bottomrule
\end{tabular}
\label{tab:complexity}
\end{table}
Highlighted by Table \ref{tab:complexity}, once the evolving probability law has been learned, the reduced dynamics act on a low-dimensional 
space and avoid any dependence on unresolved degrees of freedom.  
As a result, surrogate trajectories can be evaluated efficiently and repeatedly without recomputing the underlying high-dimensional 
dynamics. 
\section{Numerical experiments}\label{sec:numerical experiments}
\noindent In this section, we present the results of several numerical experiments that demonstrate the effectiveness of the proposed Pseudo-Markovian scheme. The numerical experiments illustrate how the proposed pseudo-Markovian reduction performs across settings of different complexity. We begin with linear–Gaussian systems, where closed-form solutions allow a direct verification of marginal preservation and conditional accuracy . We then consider nonlinear dynamics with non-Gaussian initial measures, testing the method’s ability to track probability flows in a multimodal setting. Finally, we examine SDE setting, demonstrating the performance of the method in dynamics driven by White noise. Together, these examples verify accuracy, probability consistency, and efficiency of the method across representative scenarios. \\ \ \\
We focus on systems with polynomial drift terms $b(X)$ and constant noise coefficient matrices $\sigma$.
The overall solution strategy is summarized in Algorithm \ref{alg:solution algo}:
\begin{enumerate}
    \item \textit{Offline phase}. 
    The probability flow $\{\mu_t\}$ is first approximated using a PC expansion Eq:~(\ref{eq:PC expansion}) by solving Eq.~(\ref{eq:coeff evolution ODE}) or Eq.~(\ref{eq:coeff evolution SDE}), depending on whether the underlying dynamics is deterministic or stochastic.
    
    \item \textit{Online phase}. For any chosen set of resolved variables $A(X)$, we integrate the reduced dynamics Eq.~(\ref{eq:Reduced Dynamics Discretized}) using aforward Euler (ODE), or Euler-Maruyama (SDE) scheme to approximate $\mathbb{E}[\hat{X}_t\mid\hat{X_0}]$. 
    The projected drift term is represented by the polynomial expansion Eq.~(\ref{eq:estimate_CE}), with coefficient obtained solving the linear system Eq.~(\ref{eq:linear system CE}). These coefficients depend explicitly on the instantaneous value of the PC coefficients precomputed in the offline phase.
\end{enumerate}
\noindent\textbf{Remark.}
Note that in the experiment with nonlinear dynamics Eq.~(\ref{eq:nonlinear dynamics}), the drift term does not satisfy the global-Lipschitzness assumption of Section~\ref{sec:assumptions}. However, there exists a global Lyapunov function with negative time-derivative ensuring that $b(X)$ remains Lipschitz along trajectories (they remain in a compact set).
There is moreover a single asymptotically stable fixed point. We expect the possibility to extend our results to Lyapunov stable dynamical systems, as supported by the numerical experiment.
\subsection{Linear dynamics with Gaussian initial probability measure}\label{sec:linear ode gaussian}
\noindent Consider the linear dynamics for $X\in\mathbb{R}^N$ with uncertain initial condition $X_0$,
\begin{eqnarray}\label{eq:linear system ODE}
    \frac{d}{dt}X = BX,\quad
    X_0 \sim \mathcal{N}(m_0,\Sigma_0),
\end{eqnarray}
where $B \in \mathbb{R}^{N\times N}$ is a constant matrix. 
We set $N=10$, $m=1$, and choose the resolved variables as $\hat{X}=X_1$.
The components of the initial mean $m_0\in\mathbb{R}^{10}$ are set to $m_{0,i}=1, \forall i = \{1,\dots,10\}$ and the initial covariance matrix $\Sigma_0\in\mathbb{R}^{10 \times 10}$ is given in Appendix \ref{app:numerical details}.
The components of the matrix $A$ are such that $B_{ii}=-1$ and $B_{ij}=B_{ji}=0.1, \forall i,j= \{1,\dots,10\}, i\neq j$.
\\ \ \\
In Fig.~\ref{fig:linear gaussian m=1 N=10}, the evolution of the surrogate trajectory for the resolved variable $\hat{X}=X_1$ is compared against: (i) the closed-form solution for $\mathbb{E}[\hat{X}_t\mid\hat{X}_0]$; (ii) an MC estimate with $10^5$ initial samples; and (iii) the First Order Optimal Prediction (Sec.~\ref{sec:FOP}).
The approximation error for both the surrogate trajectory and the MC estimate after 1000 time steps and for $B_{ij}=0.01$ is plotted as a function of the initial marginal probability density $\rho(\hat{X},t=0)$ with emphasis on the low probability regime. 
A similar plot is shown for the evolution of the total expectation $\mathbb{E}[\hat{X}_t]$, where the evolution of the 0-th PC coefficient is compared with the closed-form solution and the MC estimate with $10^5$ samples. 
We also show the convergence of several statistical moments after 1000 time steps, obtained either from the PC evolution or MC estimates, as a function of the number of samples. Finally, Fig.~\ref{fig:error modes} illustrates the convergence of surrogate trajectories converge to $\mathbb{E}[\hat{X}_t\mid\hat{X}_0]$ as the number $m$ of resolved variables increases, and similarly for the two-time joint density $\rho_{|\hat{X}_0}(\hat{X_t})$.
The long-time asymptotic convergence of the surrogate trajectory is also shown. 
\subsection{Error estimate}\label{sec:error estimate}
\noindent We now analyze the modeling error in $\mathbb{E}[\hat{X}_t\: | \: \hat{X}_0]$, introduced by the pseudo-Markovian approximation, Eq.~(\ref{eq:surrogate dynamics proposition}). We focus on the deterministic setting, while the stochastic case is discussed in Appendix \ref{app:error sde}. 
In the special case of a linear drift and of a Gaussian initial probability measure, the modeling error admits an explicit expression for the time evolution.
To derive it, we first write the exact evolution $\hat{X}^{exact}$ of the resolved components $\hat{X}=(X_1,\dots,X_m)\in\hat{\Gamma}\subset\mathbb{R}^m$,
\begin{eqnarray}\label{eq:exact resolved}
    \frac{d}{dt}\hat{X}^{exact}_t = \mathbb{E}_{\mu_t}[\hat{b}\mid \hat{X}^{exact}_t] + F_t,
\end{eqnarray}
with initial condition $\hat{X}^{exact}_0 = \hat{X}_0$ and
where the random force $F_t \in \mathbb{R}^m$ denotes the complement of the optimal projection that is orthogonal to the resolved drift $\hat{b}\in\mathbb{R}^m$ with respect to the conditional expectation operator  $\mathbb{E}_{\mu_t}[\:(\cdot)\mid\hat{X}_t\:]$.
\\ \ \\
The unresolved variables $\hat{X}_0^{\perp ,exact} \in\mathbb{R}^{N-m}$ are sampled from $\mu_{0|\hat{X}_0}$, the initial probability measure conditioned on $\hat{X}_0$.  
In this deterministic setting, a surrogate trajectory $\hat{X}^S_t$ is defined as the solution of Eq.~(\ref{eq:surrogate dynamics proposition}) with fixed initial condition $\hat{X}^{S}_0 = \hat{X}_0$.
It provides an approximation of the conditional expected trajectory $\mathbb{E}[\hat{X_t}\mid\hat{X}_0]$, whose exact evolution can be derived by taking the conditional expectation of Eq.~(\ref{eq:exact resolved}) with respect to $\hat{X}_0$,
\begin{eqnarray}
    \frac{d}{dt}\mathbb{E}[\hat{X}^{exact}_t\mid \hat{X}_0] = \mathbb{E}\left[\mathbb{E}_{\mu_t}[b\mid\hat{X}^{exact}_t] \:\big|\: \hat{X}_0 \right] + \mathbb{E}[F_t\mid\hat{X}_0],\quad
    \mathbb{E}[\hat{X}_0^{exact}\mid\hat{X}_0] = \hat{X}_0.
\end{eqnarray}
Defining the modeling error as $e(t)=\hat{X}^S_t-\mathbb{E}[\hat{X}^{exact}_t\mid\hat{X}_0]$, its evolution satisfies
\begin{eqnarray}
    \frac{d}{dt}e = \mathbb{E}_{\mu_t}[b\mid\hat{X}^S_t] - \mathbb{E}\left[\mathbb{E}_{\mu_t}[b\mid\hat{X}^{exact}_t] \mid \hat{X}_0 \right] - \mathbb{E}[F_t\mid\hat{X}_0],\quad
    e(0)=0.
\end{eqnarray}
For a linear system governed by Eq.~(\ref{eq:linear system ODE})
with Gaussian initial probability density $\mathcal{N}(m_0,\Sigma_0)$, the conditional expectation can be expressed explicitly, and we obtain
\begin{align}
    \frac{\mathrm{d}}{\mathrm{d}t} e(t) 
    &= Q(t)\, e(t) - \overline{F}(t), \label{eq:evolution_error} \\
    \overline{F}(t) 
    &:= \mathbb{E} [ F_t \mid \hat{X}_0 ], \label{eq:mean_force} \\
    Q(t) 
    &:= B_{\mathrm{res},\,\mathrm{res}} 
    + B_{\mathrm{res},\,\mathrm{unres}} 
    \left[ e^{Bt} \Sigma_0 e^{B^\top t} \right]_{\mathrm{unres},\,\mathrm{res}} 
    \left[ e^{Bt} \Sigma_0 e^{B^\top t} \right]^{-1}_{\mathrm{res},\,\mathrm{res}}, \label{eq:Q_def}
\end{align}
where, $m_0 \in \mathbb{R}^N$ and $\Sigma_0 \in \mathbb{R}^{N\times N}$ are the initial mean and covariance matrix, respectively.
The sub-matrix $B_{\mathrm{res},\,\mathrm{res}} \in \mathbb{R}^{m \times m}$ is obtained by selecting the rows and columns of $B$ corresponding to the resolved variables, whereas the sub-matrix $B_{\mathrm{res},\,\mathrm{unres}} \in\mathbb{R}^{m\times (N-m)}$ is obtained by selecting the rows corresponding to the resolved variables and the columns corresponding to the unresolved variables. The same subscript convention applies to the other sub-matrices appearing in Eq.~(\ref{eq:Q_def}). The matrix exponential, $e^{Bt}=\sum_{k=0}^\infty \frac{t^kA^k}{k!}\in \mathbb{R}^{N \times N}$ is assumed to yield a nonsingular sub-matrix $\left[ e^{Bt} \Sigma_0 e^{B^\top t} \right]_{\mathrm{res},\,\mathrm{res}}$  on $[0,T]$. Under these assumptions, the unique solution of Eq.~(\ref{eq:evolution_error}) is
\begin{eqnarray}\label{eq:error analytical}
    e(t) = -\int_0^t \Phi(t,\tau) \overline{F}(\tau)  d\tau, \ \ t\in[0,T],
\end{eqnarray}
where we defined the time-ordered exponential,
\begin{eqnarray}
    \Phi(t,\tau) = \mathcal{T}_{exp}\left(  \int_\tau^t Q(s)ds \right).
\end{eqnarray}
Figure \ref{fig:error modes} illustrates that, as the number $m$ of resolved variables increases, the surrogate trajectory $\hat{X}^S_{1,t}$ converges exponentially to $\mathbb{E}[\hat{X}_{1,t}^{exact}\mid\hat{X}_{1,0}]$, where $\hat{X}^S_{1,t}$, $\hat{X}_{1,t}^{exact}$ and $\hat{X}_{1,0}$ denote the first components of $\hat{X}_t^S$, $\hat{X}^{exact}_t$, and $\hat{X}_0$, respectively. 
When all variables of the original system are resolved, $m=N$, the surrogate trajectories are exact by construction.

\begin{figure}[htbp]
    \centering
    \begin{subfigure}[b]{0.32\textwidth}
        \includegraphics[width=\textwidth]{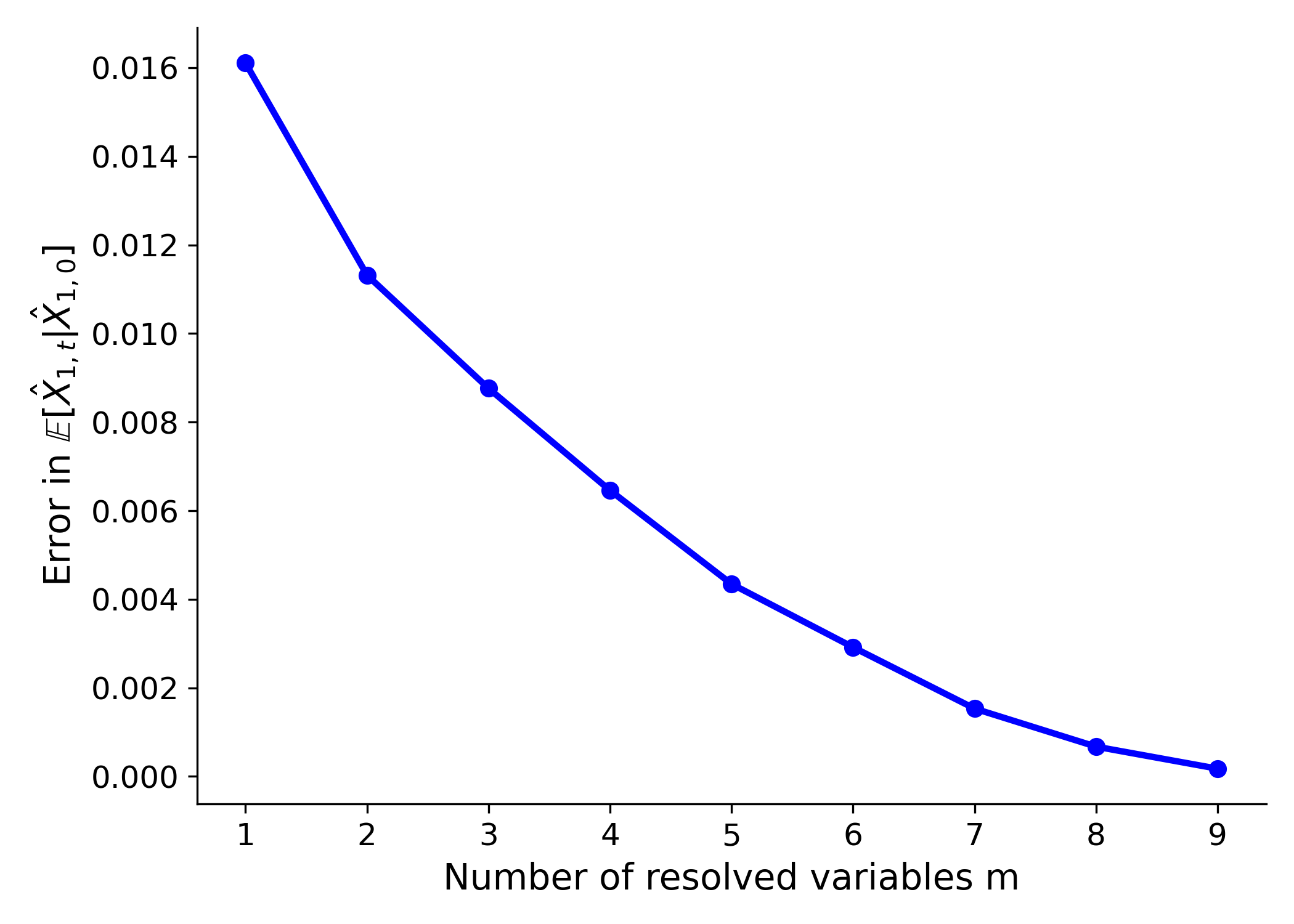}
        \caption{}
    \end{subfigure}
    \begin{subfigure}[b]{0.32\textwidth}
        \includegraphics[width=\textwidth]{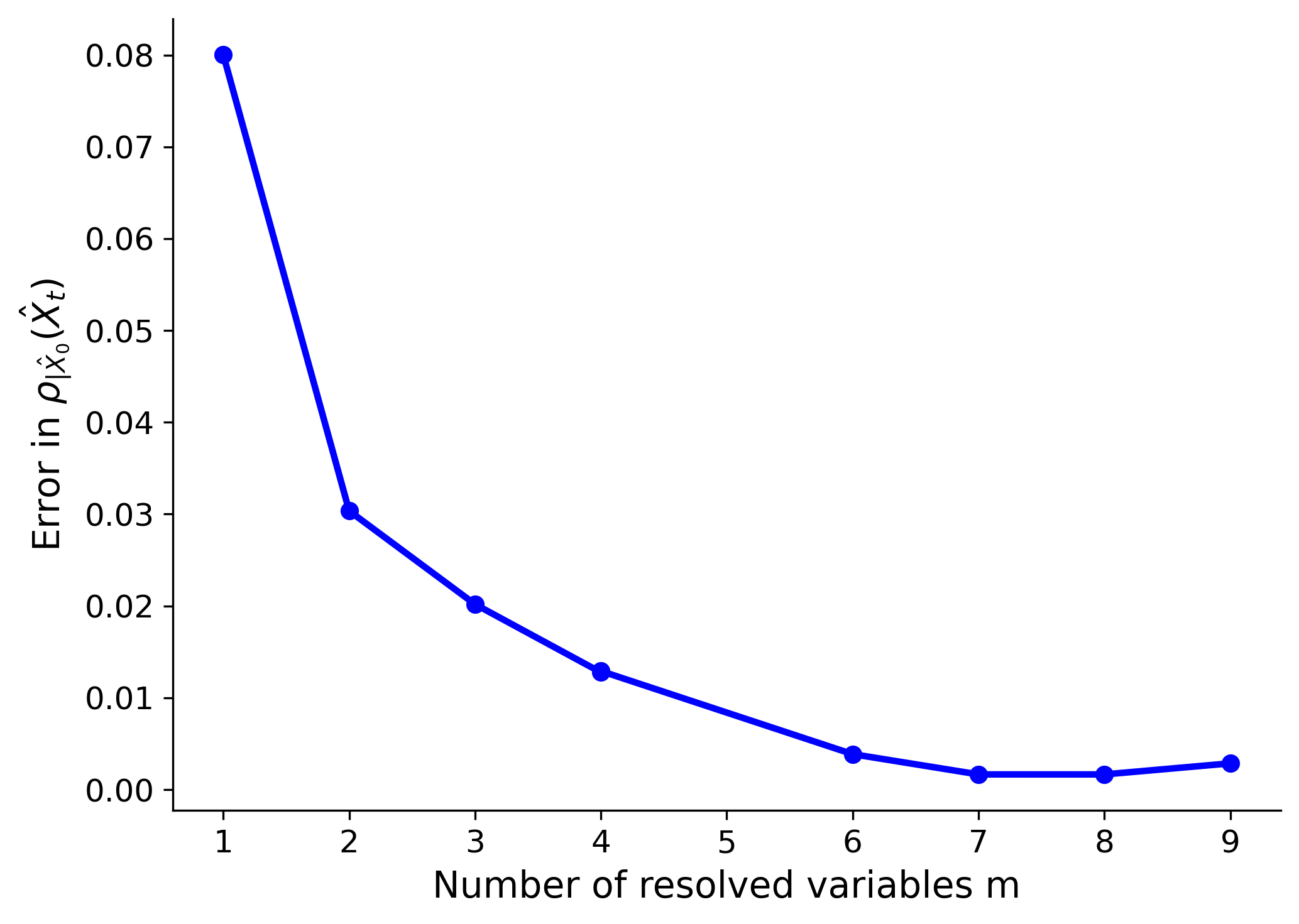}
        \caption{}
    \end{subfigure}
    \begin{subfigure}[b]{0.32\textwidth}
        \includegraphics[width=\textwidth]{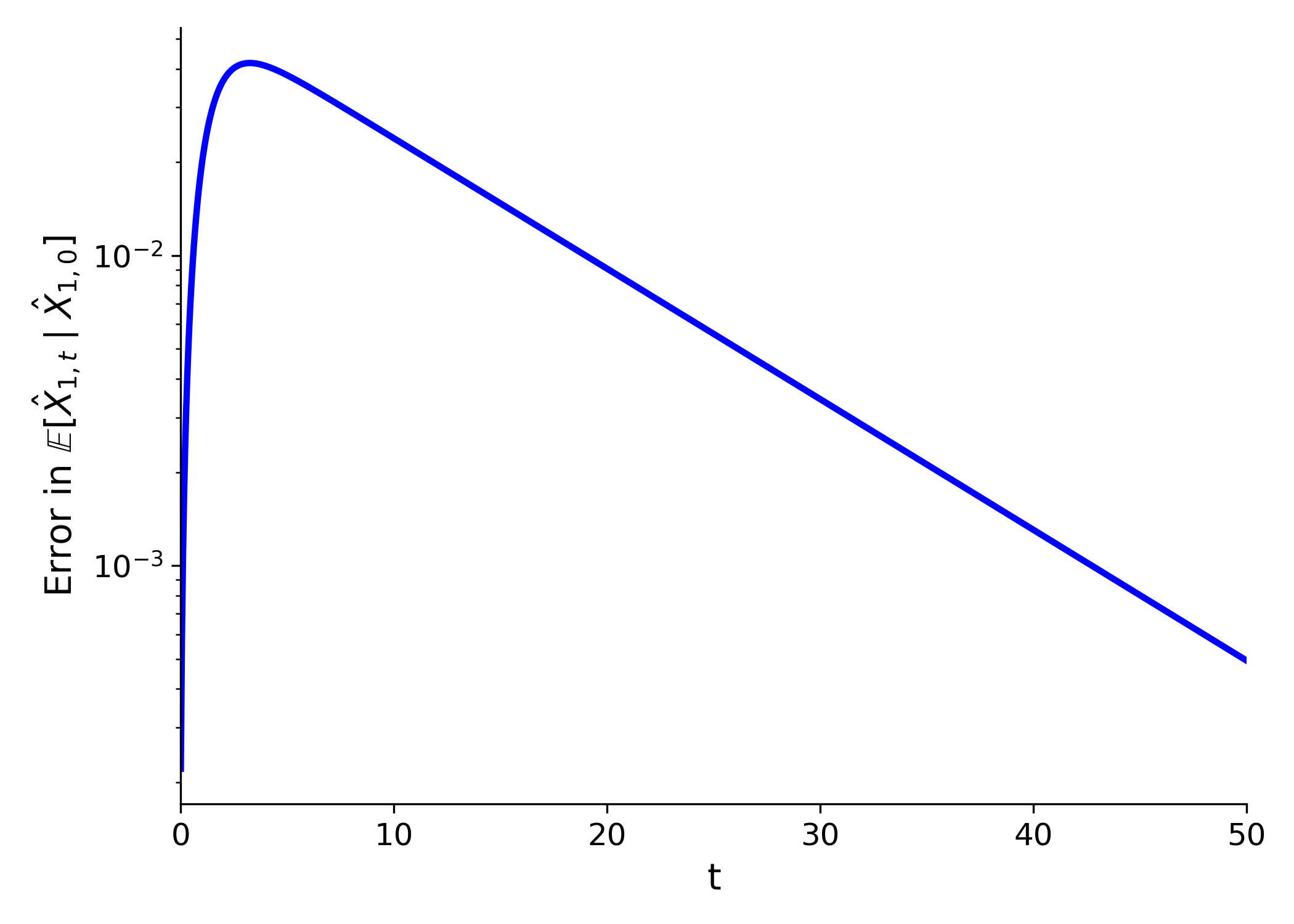}
        \caption{}
    \end{subfigure}
    
    \caption{\justifying 10D linear dynamics. (a) $L^2$ error in $\mathbb{E}[\hat{X}_{1,T}\mid\hat{X}_{1,0}]$ as a function of the number of resolved variables $m$. The error was computed using the analytical expression Eq.~(\ref{eq:error analytical}).
    (b) $W_2$ error in the two-time probability density $\rho_{|\hat{X}_0}(\hat{X}_T)$ when $m=1$ and $\hat{X}=X_1$.
    (c) $L^2$ error in $\mathbb{E}[\hat{X}_T\mid\hat{X}_0]$ as a function of time when $m=1$ and $\hat{X}=X_1$.}
    \label{fig:error modes}
\end{figure} 

\subsubsection{Short times limit}
\noindent For short times, Eq.~(\ref{eq:error analytical}) reduces to (see Appendix \ref{app:error}),
\begin{eqnarray}\label{eq:error short}
    e(t) =  -\frac{t^2}{2}\overline{F}'(0) + O(t^3),
\end{eqnarray}
where $\overline{F}'(0)$ is the time derivative of $\overline{F}$ at initial time. In Appendix \ref{app:error}, we derive an explicit expression for a two dimensional system.
Equation~(\ref{eq:error short}) is particularly useful to assess whether additional resolved variables are required to reach a prescribed accuracy over a short time horizon. 

\subsubsection{Stochastic dynamics}
\noindent In Appendix \ref{app:error sde}, we show that for a dynamical system governed by a stochastic evolution law with linear drift and a constant diagonal noise coefficient matrix, the error evolution follows an expression similar to Eq.~(\ref{eq:error analytical}), but with a different definition of $Q(t)$, Eq.~(\ref{eq:Q_def}).

\subsection{Nonlinear dynamics with bimodal initial probability measure}
\noindent We consider the 10-dimensional non-linear deterministic dynamical system for $X\in\mathbb{R}^{10}$,
\begin{eqnarray}\label{eq:nonlinear dynamics}
    \frac{d}{dt}X_i = a_iX_i + \sum_{j=1}^{10}K_{ij}X_j^3,
    \quad i=\{1,\dots,10\}, \quad
    X_0 \sim \mu_0,
\end{eqnarray}
with $a_i=-0.1, \ K_{ij}=-0.1 \ \mathrm{ when }\ j=i $, and $K_{ij}=0.01 \ \mathrm{ when }\ j\neq i, \ $ $i,j=\{1,\dots,10\}$.
We assume $m=1$ and $\hat{X}=X_1$ for resolved variable. 
Let $\nu$ be a probability measure over $X_1 \times X_2$ and $\pi$ a probability measure over $X_3\times\dots\times X_{10}$.
The initial probability measure $\mu_0$ is constructed as the product measure $\mu_0=\nu \otimes \pi$, where $\nu$ is a mixture of Gaussian distributions,
\begin{eqnarray}\label{eq:bimodal measure}
    \nu = \alpha\mathcal{N}(m_{L},\Sigma_{L}) + (1-\alpha)\mathcal{N}(m_{R},\Sigma_{R}),
\end{eqnarray}
where $\alpha=0.2$, $m_L=(-0.3,-0.3)^T$, $m_R=(0.5,0.5)^T$, and $\Sigma_{L,11}=2.5, \: \Sigma_{L,22}=2,\: \Sigma_{L,12}=\Sigma_{L,21}=0.4$, and $\Sigma_{R,11}=0.6, \: \Sigma_{R,22}=0.4,\: \Sigma_{R,12}=\Sigma_{R,21}=0.1$. 
The measure $\pi$ is a multivariate Gaussian. 
Since $\mu_0$ is non-Gaussian, fitting the initial coefficients of the PC expansion is nontrivial.
Therefore, we devise a data-driven approach based on optimal transport to optimally compute the corresponding coefficients (see \ref{app:bimodal} for more details).
The PC expansion includes multivariate Hermite polynomials up to total order 3, and the conditional expectation of the resolved component of the drift is modeled as a $x_1$-polynomial of degree 3. 
Fig.~\ref{fig:non linear ode bimodal} shows the resulting approximation for the marginal probability density $\rho(x_1)$ and $\rho(x_2)$. The remaining marginals are recovered exactly as they consist of Gaussian densities.
\\ \ \\
Fig.~\ref{fig:10Dnonlinear} show the evolution of surrogate trajectories, total expectation, and snapshots of the marginal probability densities $\rho(x_1,t)$, $\rho(x_2,t)$ and $\rho(x_3,t)$. 
\\ \ \\
In Fig.~\ref{fig:low proba ode}, we show the approximation error of both
our surrogate trajectory and the MC estimate after 100 time steps as a function of the initial marginal probability density $\rho(\hat{X}, t=0)$. In this case, the full order initial joint probability measure $\mu_0$ is set to a multivariate Gaussian distribution. The error is computed relative to a reference MC simulation with $10^6$ samples. We focus on a low probability regime.    

\begin{figure*}[htbp]
    \centering
    \includegraphics[width=\textwidth]{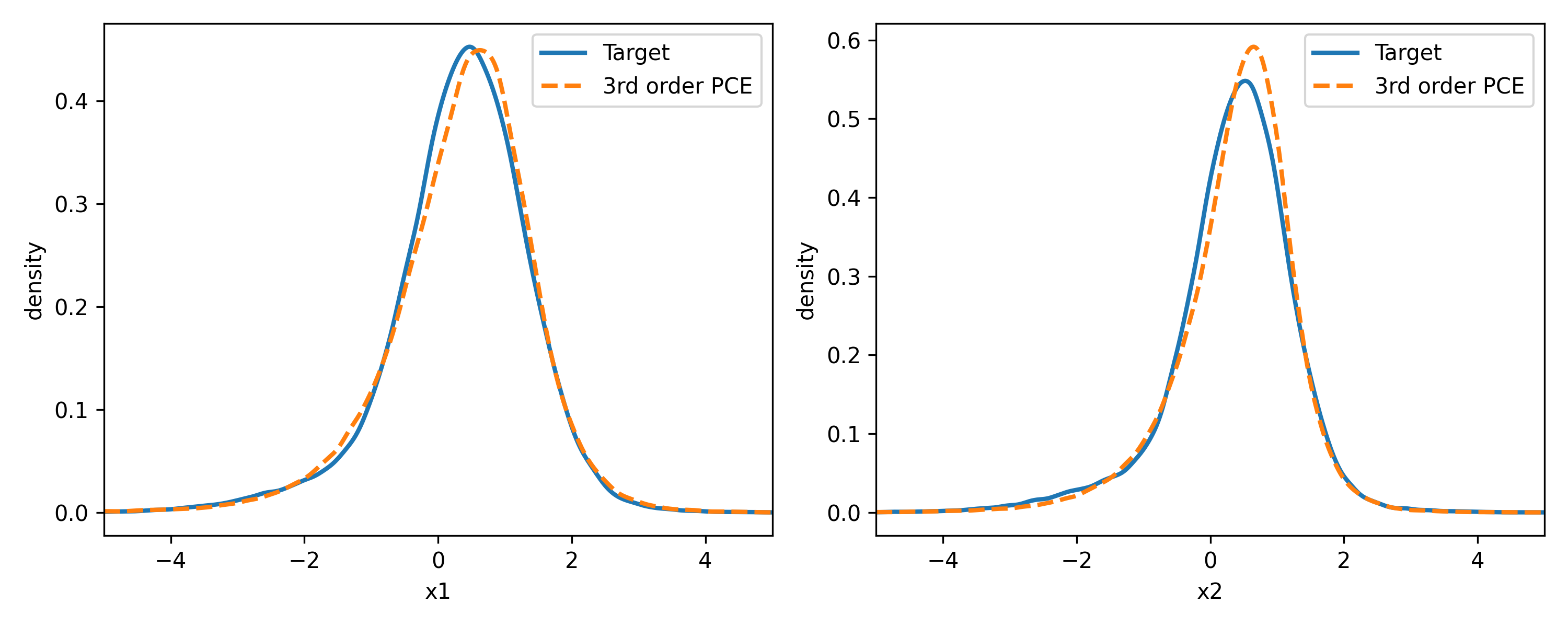}
    \caption{\justifying Initial probability density approximated with 3rd order PCE. Marginals $\rho(x_1)$ (left) and $\rho(x_2)$ (right).}
    \label{fig:non linear ode bimodal}
\end{figure*}

\subsection{Linear stochastic dynamics}\label{sec:linear sde}
\noindent We consider the stochastic process $\{X_t\}_{t\leq T}$ for the $\mathbb{R}^{10}$-valued random variable $X_t$ with linear drift $b(X)=BX\in\mathbb{R}^{10}, \ B\in\mathbb{R}^{10\times 10}$, and constant noise coefficient matrix $\sigma=\sigma_1I\in\mathbb{R}^{10\times 10}$,
\begin{eqnarray}
    dX_t = BX_t + \sigma dW_t
    ,\quad
    X_0\sim \mathcal{N}(m_0,\Sigma_0),
\end{eqnarray}
where $dW_t$ is a 10 dimensional Wiener process. 
The parameters of the initial Gaussian probability measure are $m_{0,i}=1, \forall i = \{1,\dots,10\}$, and $\Sigma_0$ is given in Appendix \ref{app:numerical details}.
We set $m=1$ and define the resolved variable as $\hat{X}=X_1$. The diagonal elements of $B$ are set to $B_{ii}=-1, \ i=\{1,\dots,10\}$.  
\\ \ \\
The results are presented in Figure \ref{fig:10D linear sde}. For different values of $\sigma_1$ and the off-diagonal elements of $B$, we compare the paths $\mathbb{E}[\hat{X}_t\mid\hat{X}_0, \ \{W_t\}_{t\leq T}]$ with fixed realizations of the Wiener process, obtained both from our surrogate trajectories and MC estimates. The approximation error in $\mathbb{E}[\hat{X}_t\mid\hat{X}_0]$ after 1000 time steps is shown as a function of the initial marginal probability density $\rho(\hat{X}, t=0)$, with an emphasis on the low-probability regime. 
The convergence of several total order statistics after 1000 time steps is displayed as a function of the number of samples, computed btoh from MC estimates and from the PC coefficients. 
Finally, the evolution of the joint probability density $\rho(x_1,x_2)$, computed from the PC, is compared to its closed form solution.
\subsection{Discussion}
\noindent In addition to preserving the marginal distribution of the resolved variables (Fig.~\ref{fig:linear gaussian m=1 N=10}), our numerical experiments indicate that the surrogate trajectories provide a meaningful approximation of $\mathbb{E}[\hat{X}_t\mid\hat{X}_0]$, even in nonlinear or non-Gaussian settings.
In the stochastic case, the reduced dynamics approximates the conditional evolution along individual realizations of the Wiener process, $\mathbb{E}[\hat{X}_t\mid\hat{X}_0, \ \{W_{\tau\leq t}\}]$.
The approximation is particularly accurate when the coupling between the resolved and unresolved variables is weak, as expected in the presence of a clear time-scale separation between the resolved and unresolved degrees of freedom. This behavior can be attributed to the use of instantaneous optimal projections.\\ \ \\
Furthermore, the empirical result shown in Fig.~(\ref{fig:error modes}) illustrates exponential convergence  with respect to the number of resolved variables for both $\mathbb{E}[\hat{X}_t\mid\hat{X}_0]$ and $\rho_{|\hat{X}_0}(\hat{X}_t)$. In more general settings, similarly fast convergence is expected due to the time-dependent optimal orthogonal projections. In addition, the total-order statistics are recovered accurately and deterministically from the PC coefficients, without resorting to sampling. \\ \ \\
Finally, the above results illustrate the ability of our method to handle rare events. With respect to standard MC methods for estimating statistics conditioned on rare events, the results highlight two advantages of our method for a fixed initial number of MC samples.
First, the accuracy of our estimates deteriorates more slowly as the probability of the conditioning event decreases. Second, our method does not rely on sampling, unlike MC, which becomes extremely inefficient in the rare-event regime, as reflected by the characteristic $-1/2$ slope in the logarithmic error plot. 


\begin{figure*}[htbp]
    \centering
    \begin{subfigure}[b]{0.49\textwidth}
        \includegraphics[width=\textwidth]{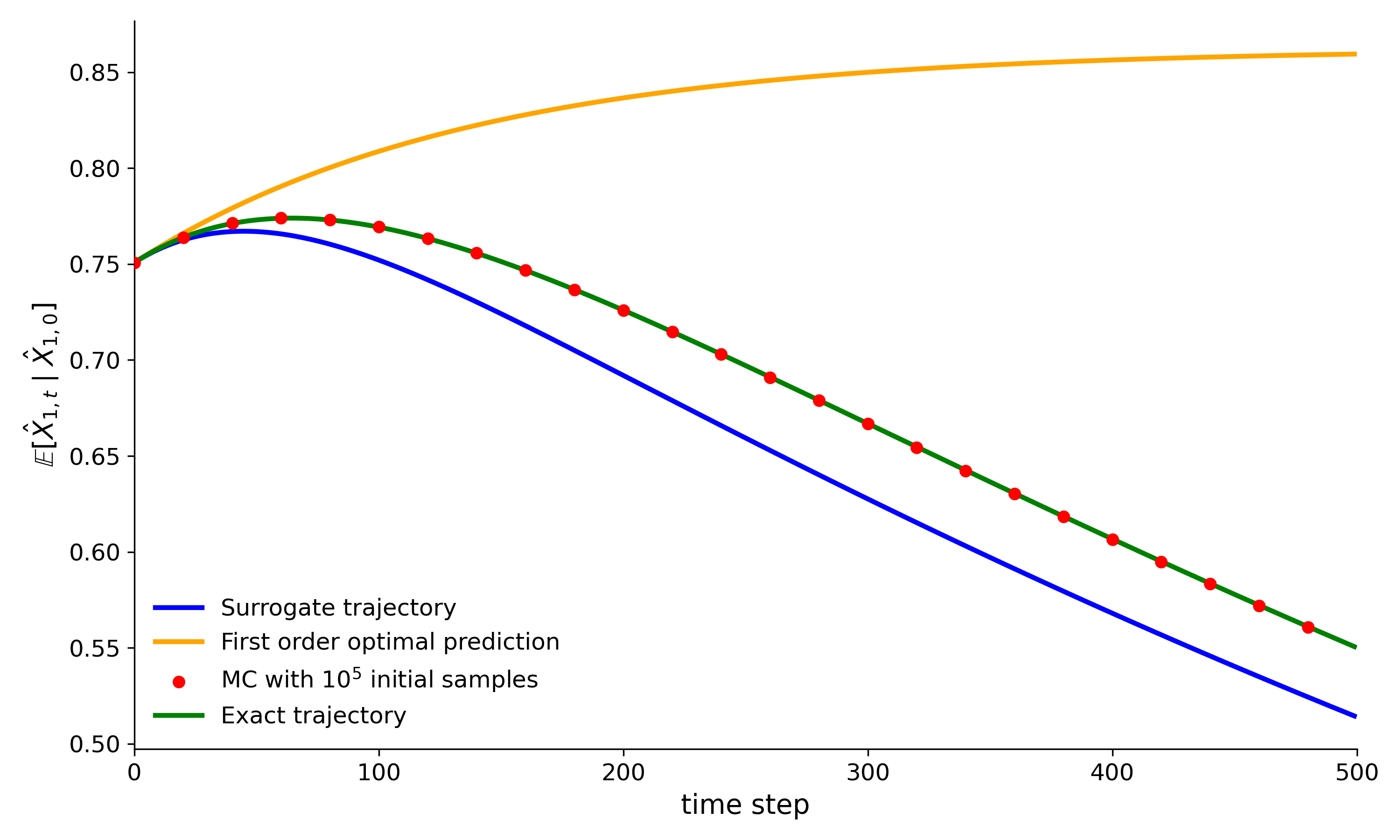}
        \caption{}
    \end{subfigure}
    \hfill
    \begin{subfigure}[b]{0.49\textwidth}
        \includegraphics[width=\textwidth]{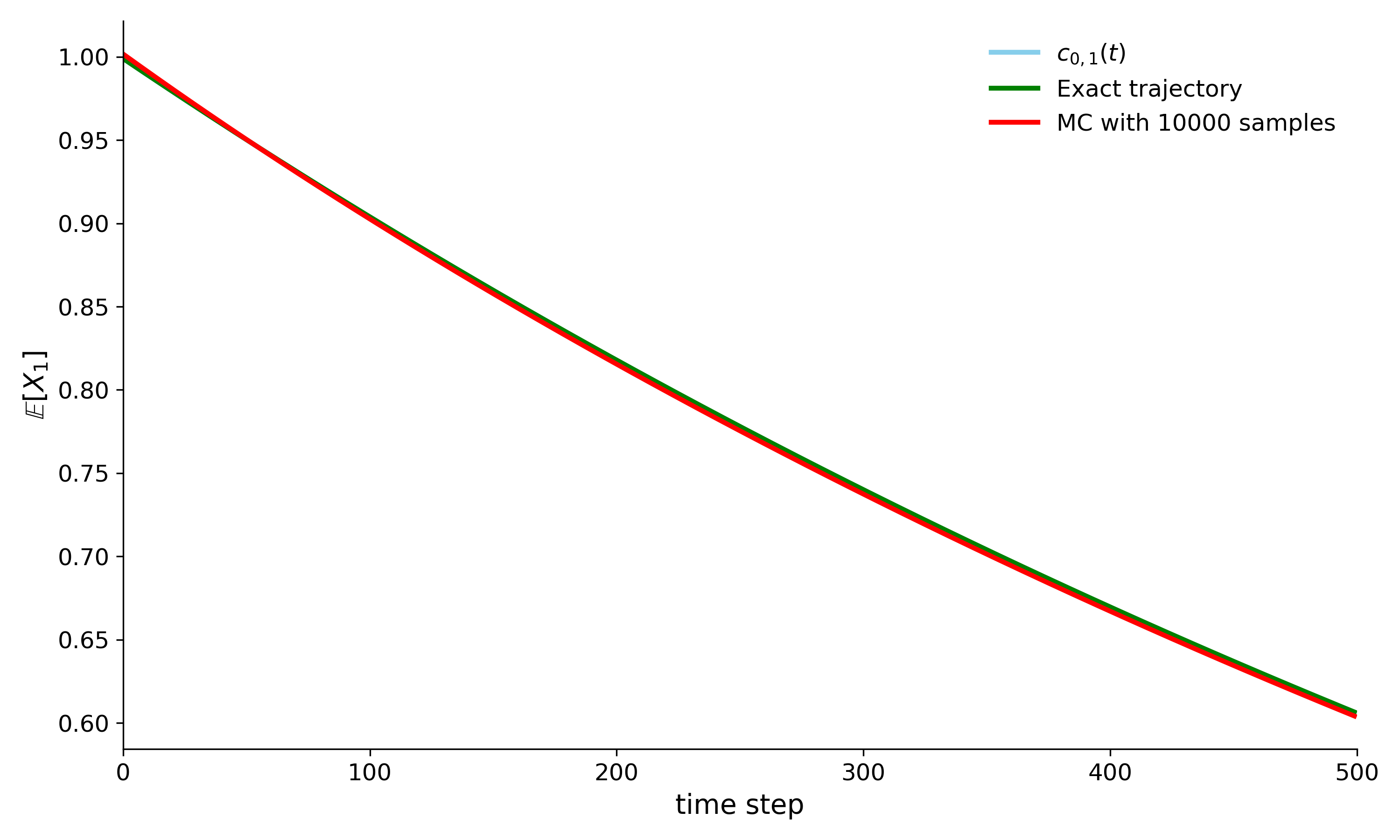}
        \caption{}
    \end{subfigure}
    \begin{subfigure}[b]{0.49\textwidth}
        \includegraphics[width=\textwidth]{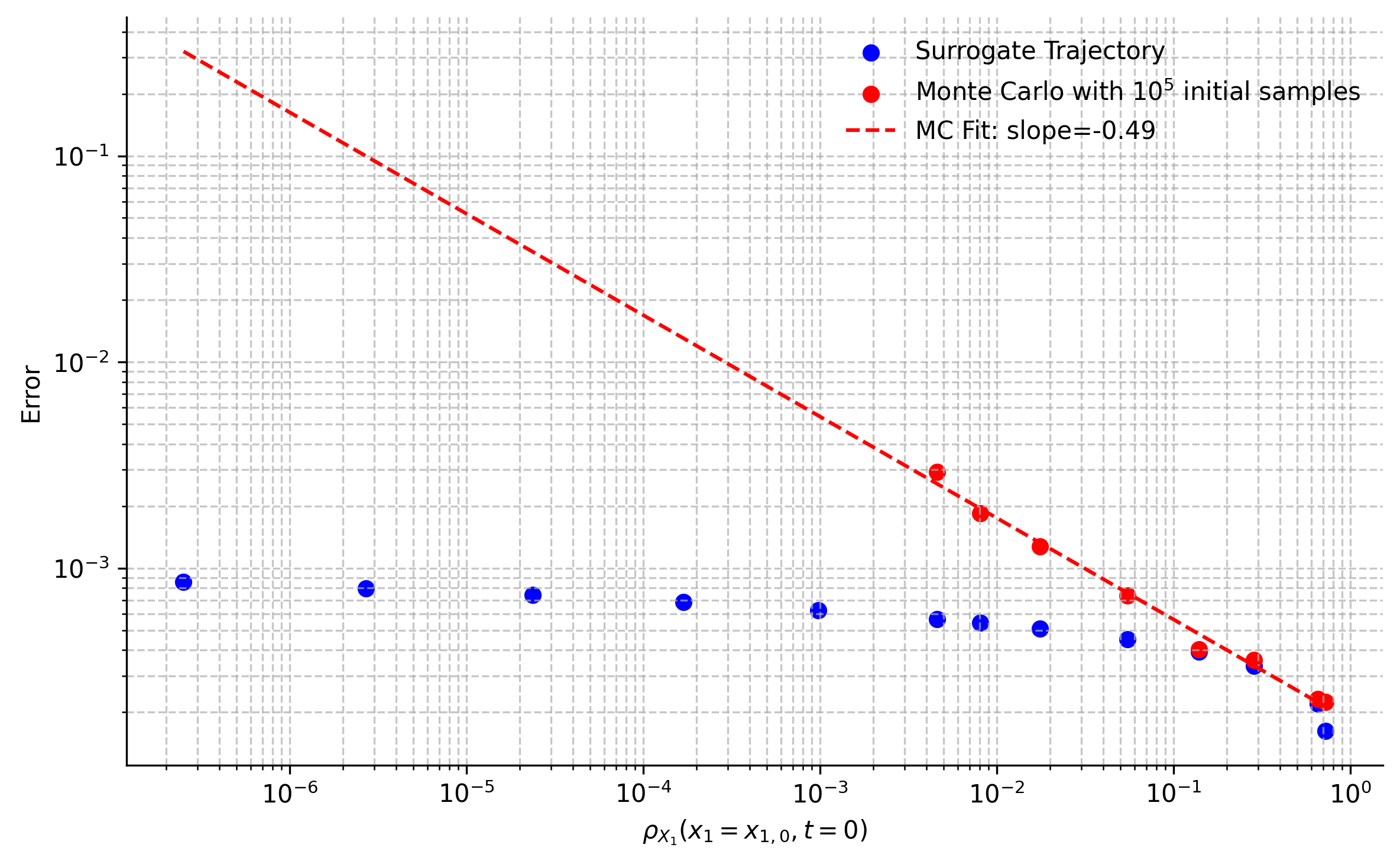}
        \caption{}
    \end{subfigure}
    \hfill
    \begin{subfigure}[b]{0.49\textwidth}
        \includegraphics[width=\textwidth]{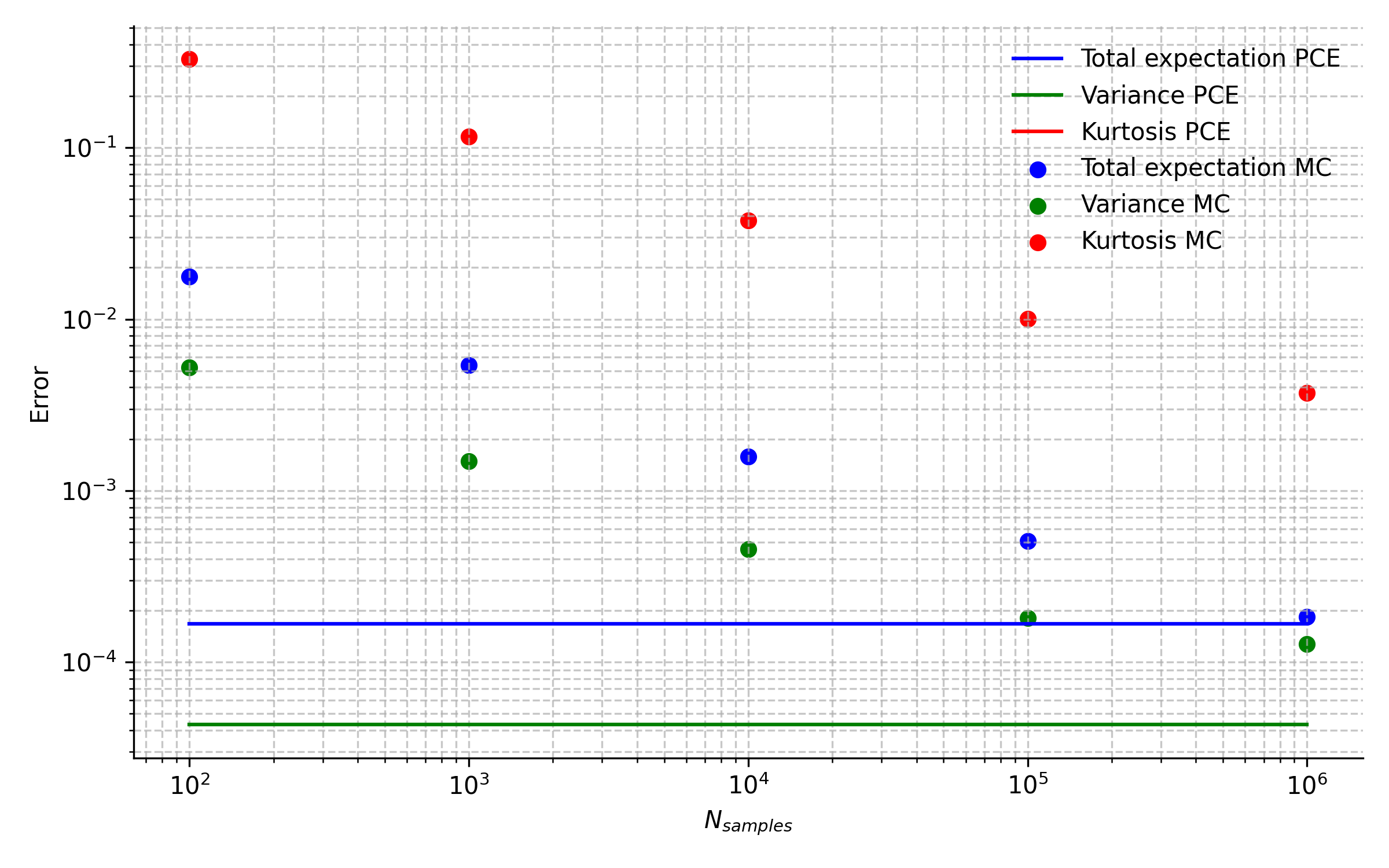}
        \caption{}
    \end{subfigure}
    
    \caption{\justifying 10D linear dynamic with initial Gaussian probability measure. The resolved variables is $\hat{X}=X_1$. The final time corresponding to time step 100 is denoted by $T$.
    (a) Evolution of $\mathbb{E}[\hat{X}_t|\hat{X}_0]$. (b) Evolution of $\mathbb{E}[\hat{X}_t]$.
    (c) Error in $\mathbb{E}[X_{1,T} \mid X_{1,0}]$ with respect to the initial marginal probability density $\rho(\hat{X}_0)$.
    (d) Mean, variance and kurtosis of $\hat{X}_{T}$  calculated with MC (dots) and PC (lines) as a function of the number of samples. Note that the Kurtosis obtained with PC is not visible because the committed error lies below $10^{-4}$.} 
    \label{fig:linear gaussian m=1 N=10}
\end{figure*}

\begin{figure*}[htbp]
    \centering
    \begin{subfigure}[b]{0.65\textwidth}
        \includegraphics[width=\textwidth]{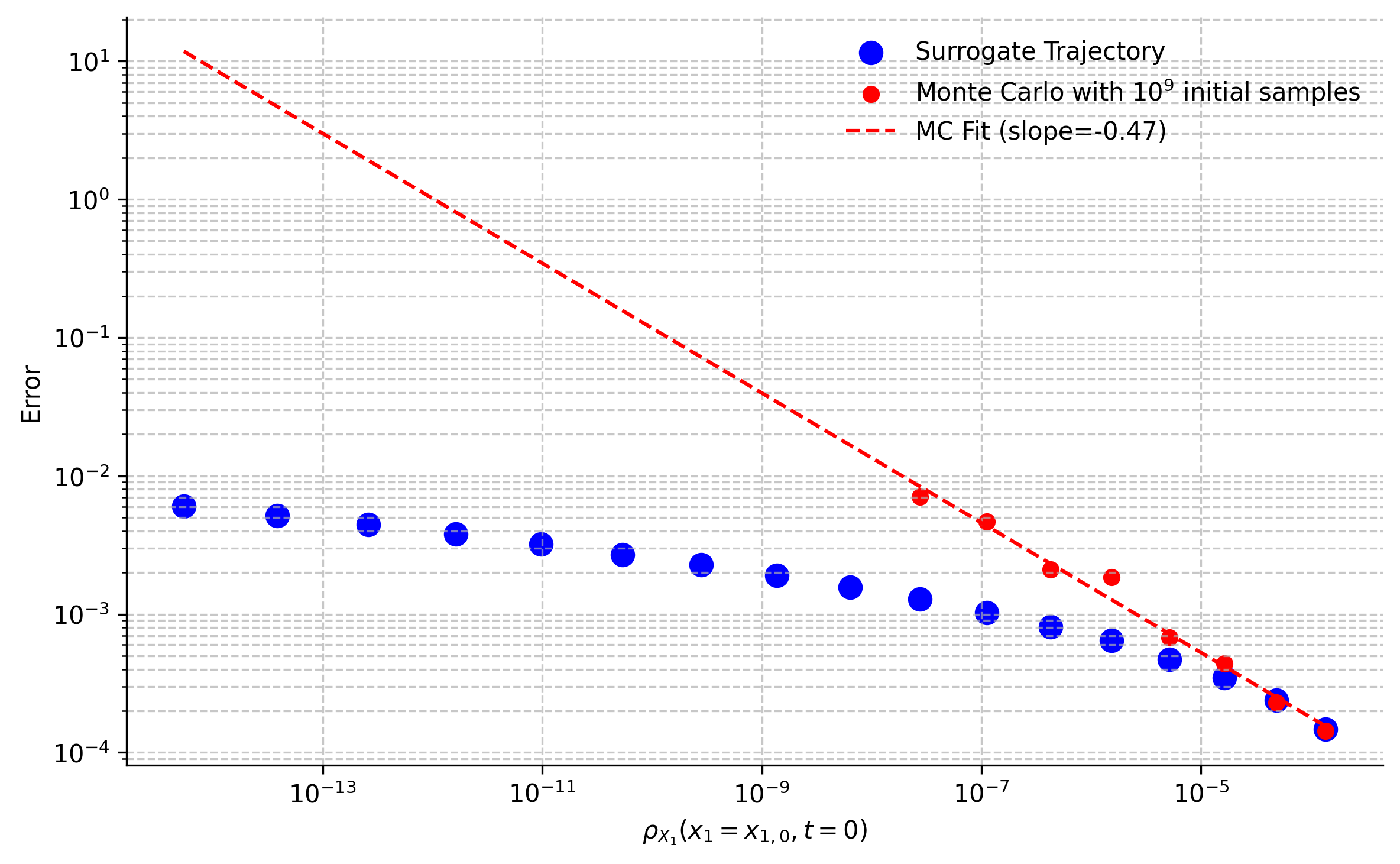}
        \caption{}
    \end{subfigure}

    \caption{\justifying 10D nonlinear system with initial Gaussian probability measure. Error in $\mathbb{E}[\hat{X}_T\mid\hat{X}_0]$ with respect to the initial marginal probability density $\rho(\hat{X}_0)$.
    }
    \label{fig:low proba ode}
\end{figure*}


\begin{figure*}[htbp]
    \centering
    \begin{subfigure}[b]{0.49\textwidth}
        \includegraphics[width=\textwidth]{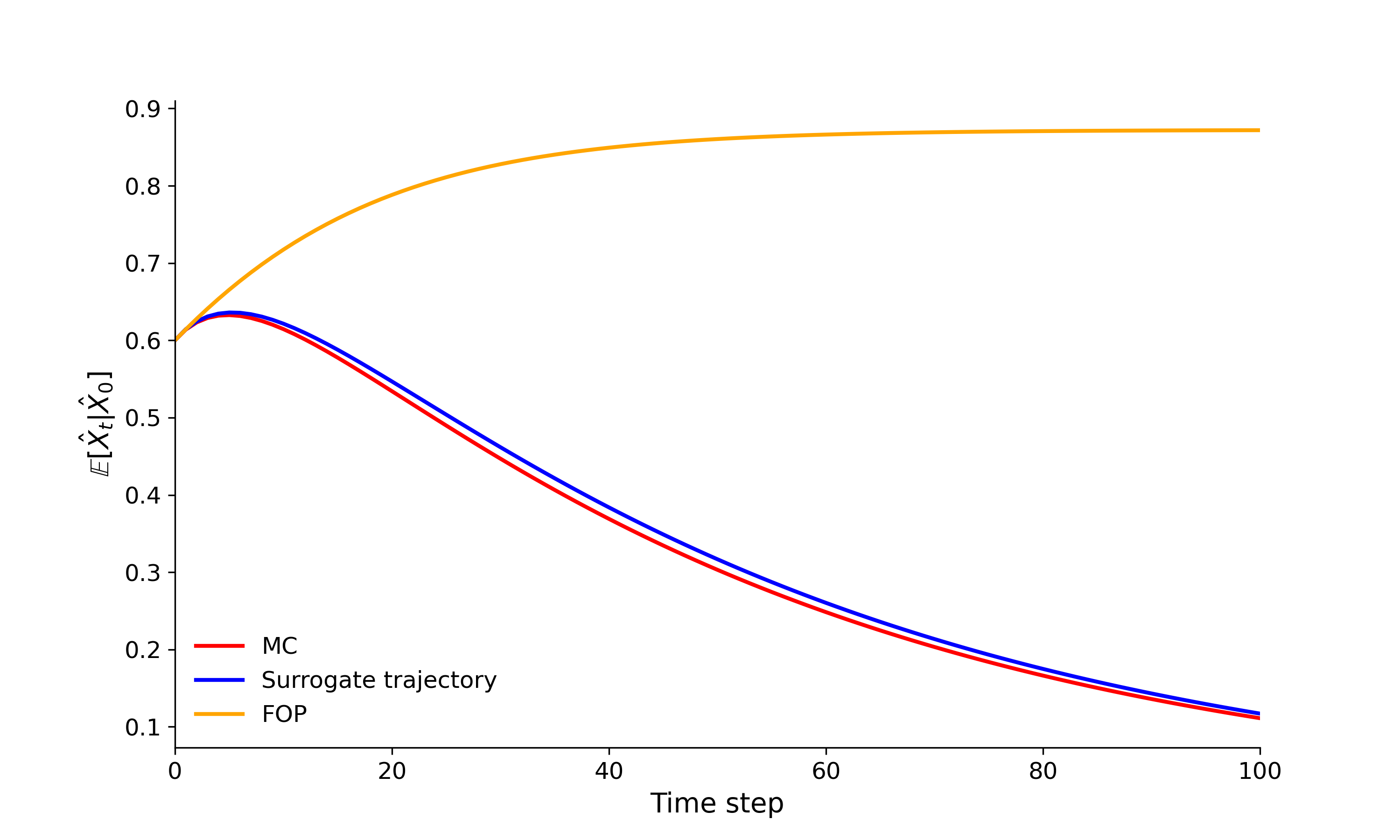}
        \caption{}
    \end{subfigure}
    \hspace{0.01cm}
    \begin{subfigure}[b]{0.49\textwidth}
        \includegraphics[width=\textwidth]{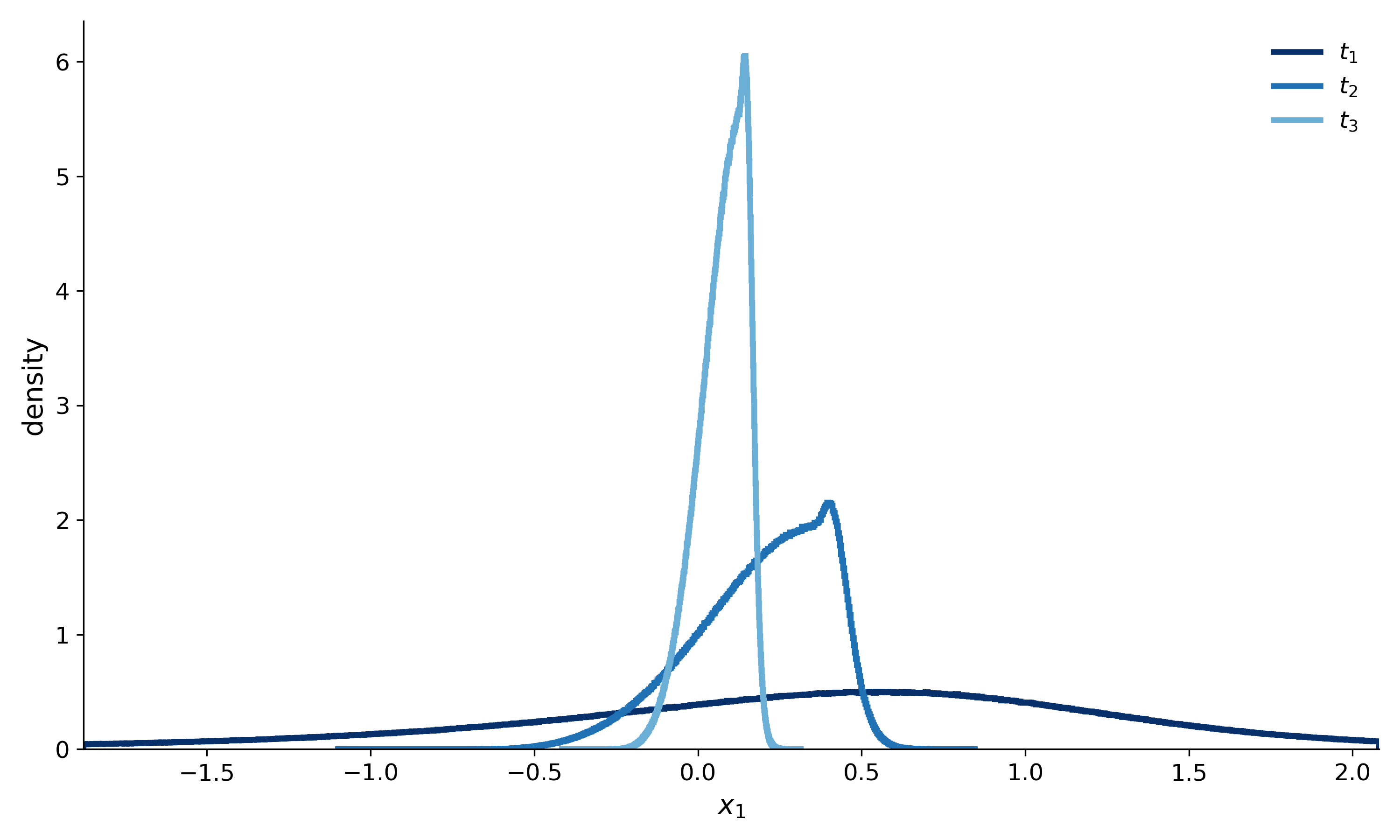}
       \caption{}
    \end{subfigure}
    \vspace{0.5cm}
    \begin{subfigure}[b]{0.49\textwidth}
        \includegraphics[width=\textwidth]{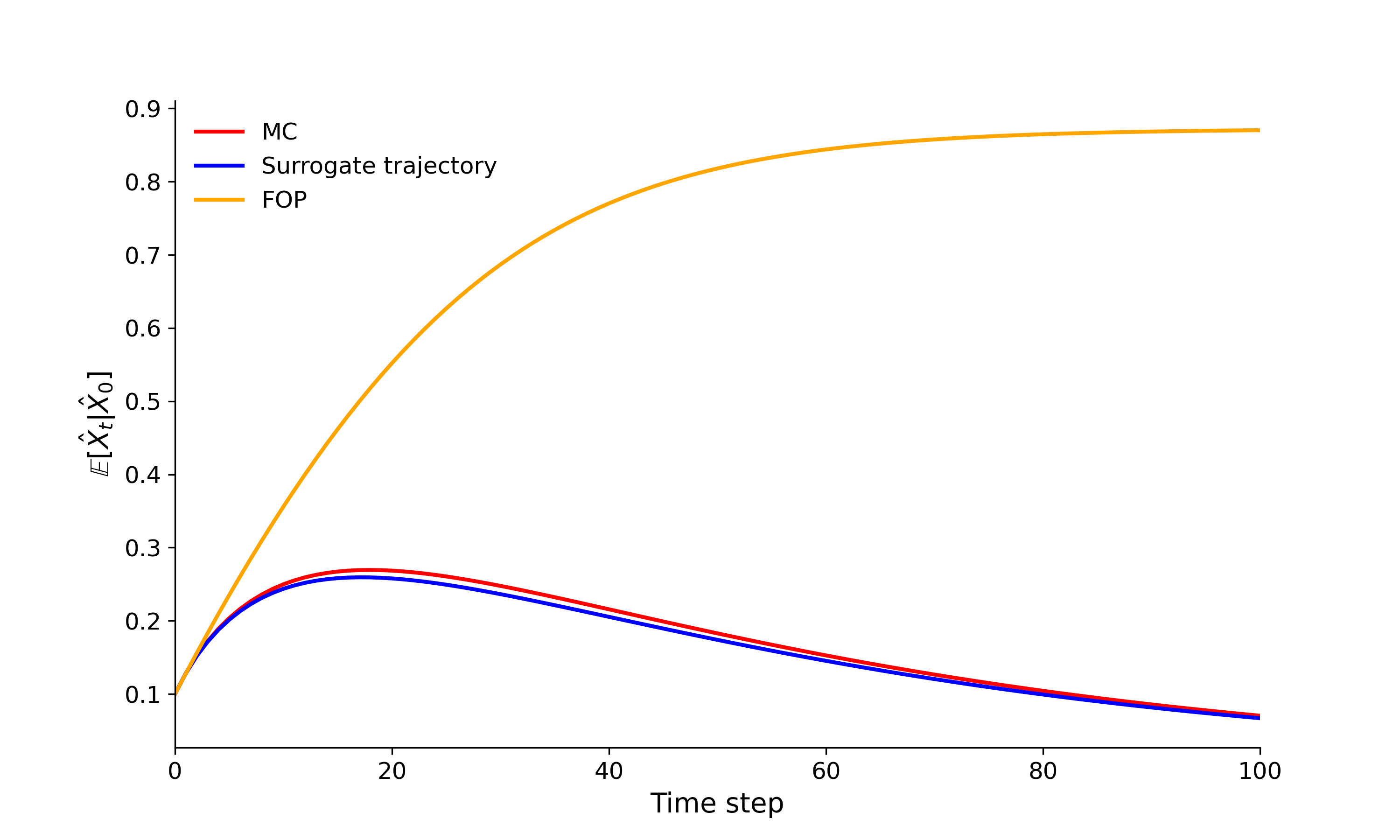}
        \caption{}
    \end{subfigure}
    \hspace{0.01cm}
    \begin{subfigure}[b]{0.49\textwidth}
        \includegraphics[width=\textwidth]{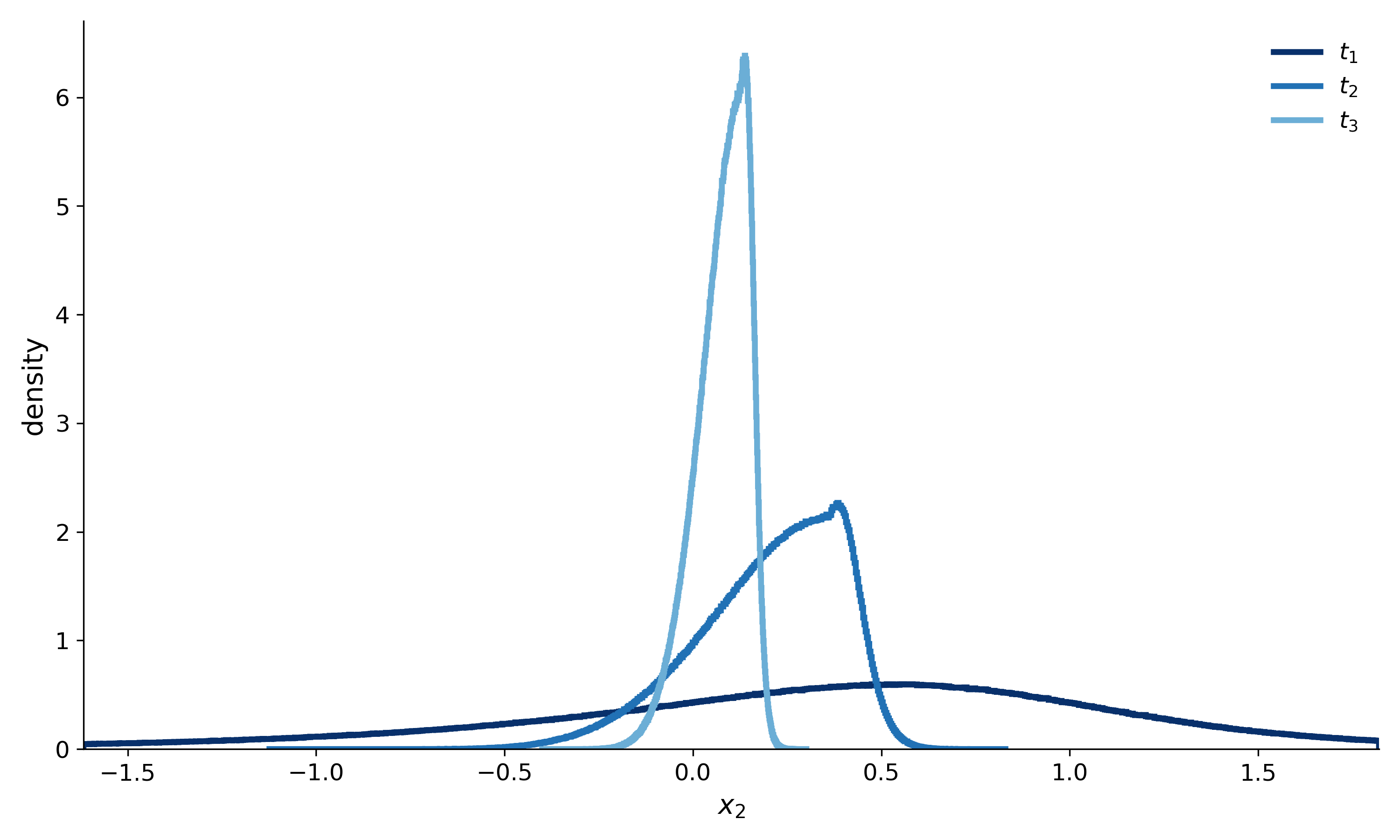}
        \caption{}
    \end{subfigure}
    \vspace{0.5cm}
    \begin{subfigure}[b]{0.49\textwidth}
        \includegraphics[width=\textwidth]{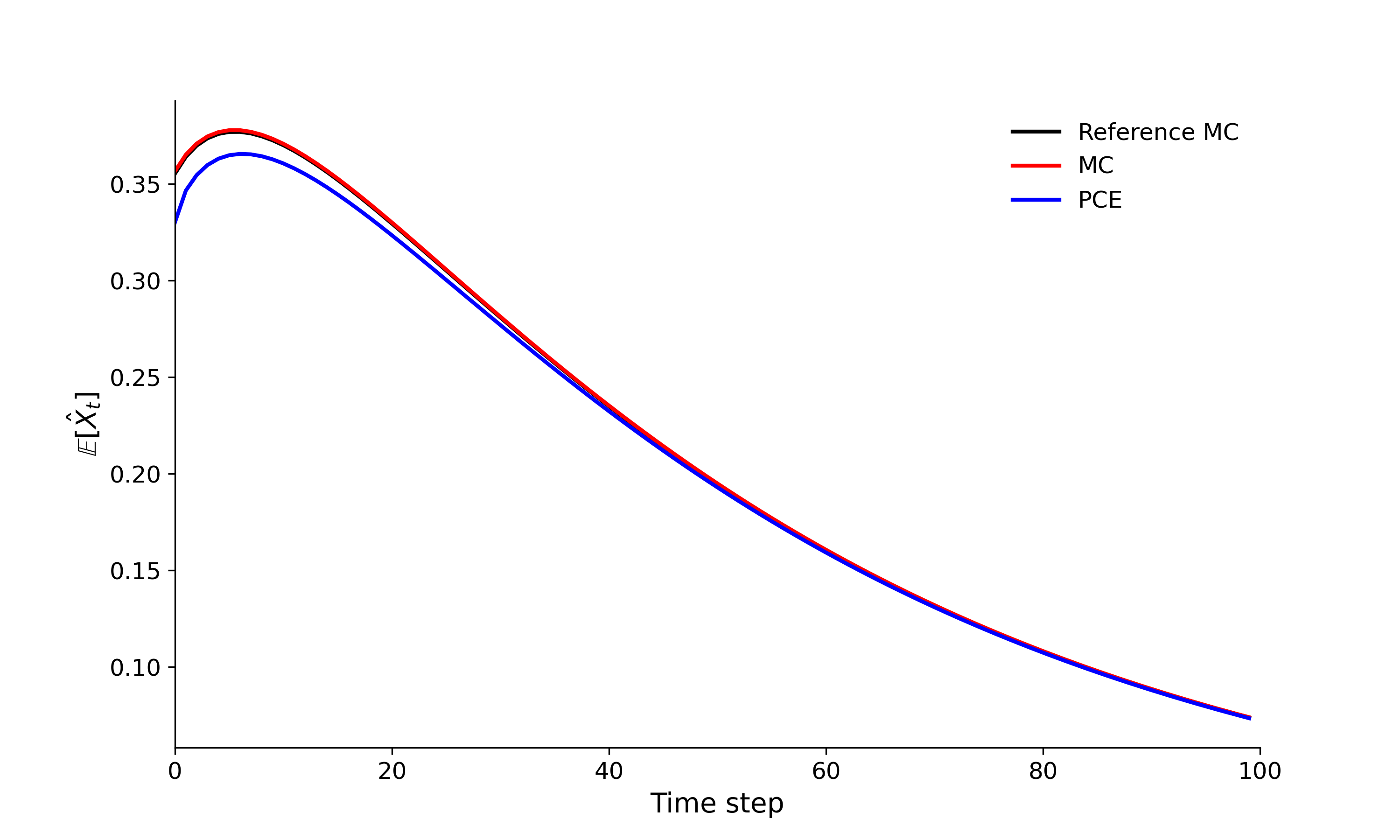}
        \caption{}
    \end{subfigure}
    \hspace{0.01cm}
    \begin{subfigure}[b]{0.49\textwidth}
        \includegraphics[width=\textwidth]{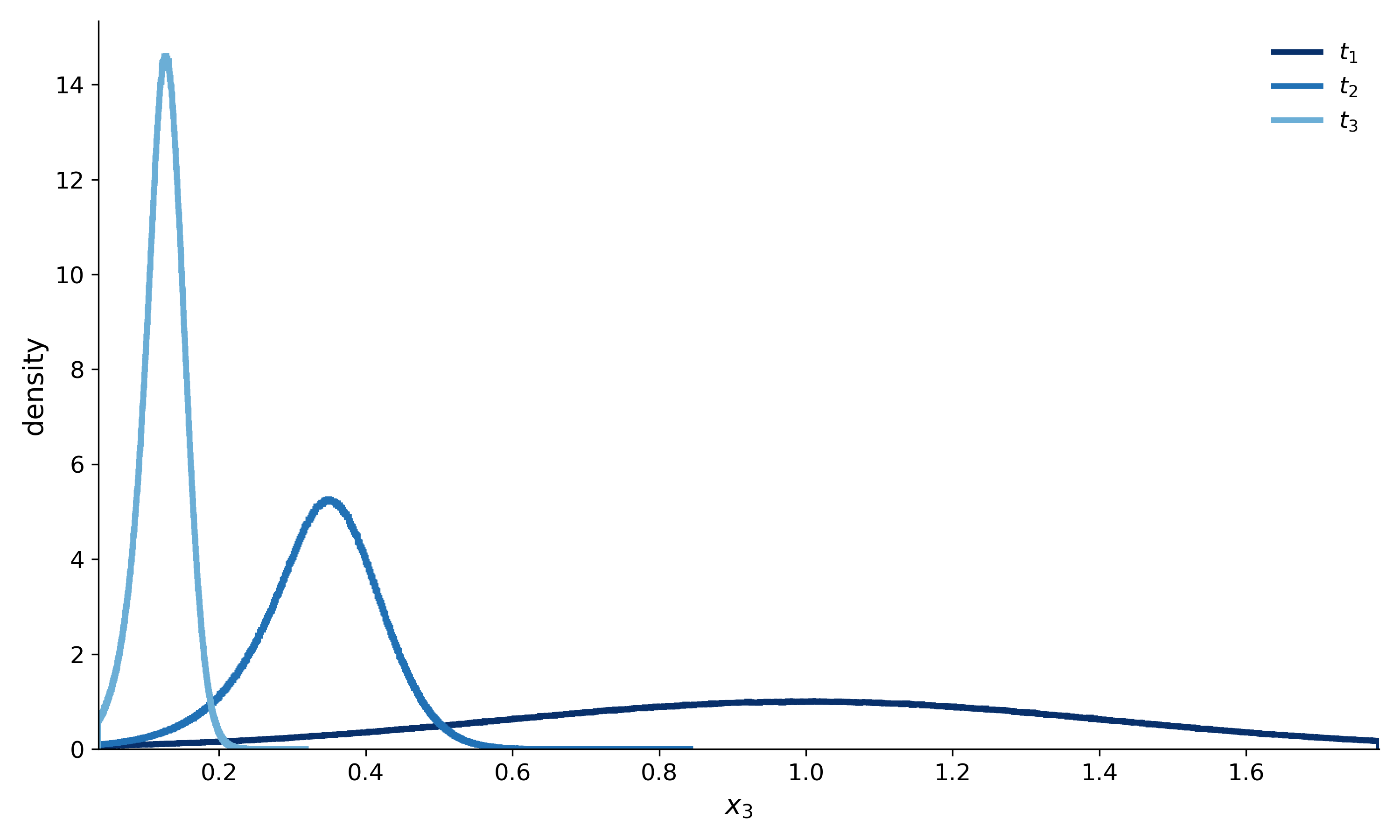}
        \caption{}
    \end{subfigure}

    \caption{\justifying 10D nonlinear dynamics with non-Gaussian initial measure.
    (a)-(b) Evolution of $\mathbb{E}[\hat{X}_t\mid\hat{X}_0]$.
    (c) Evolution of $\mathbb{E}[\hat{X}_t]$.
    (d)-(f) Evolution of the marginal probability densities $\rho(x_1,t)$, $\rho(x_2,t)$ and $\rho(x_3,t)$.
    }
    \label{fig:10Dnonlinear}
\end{figure*}

\begin{figure*}[htbp]
    \centering
    \begin{subfigure}[b]{0.88\textwidth}
        \includegraphics[width=\textwidth]{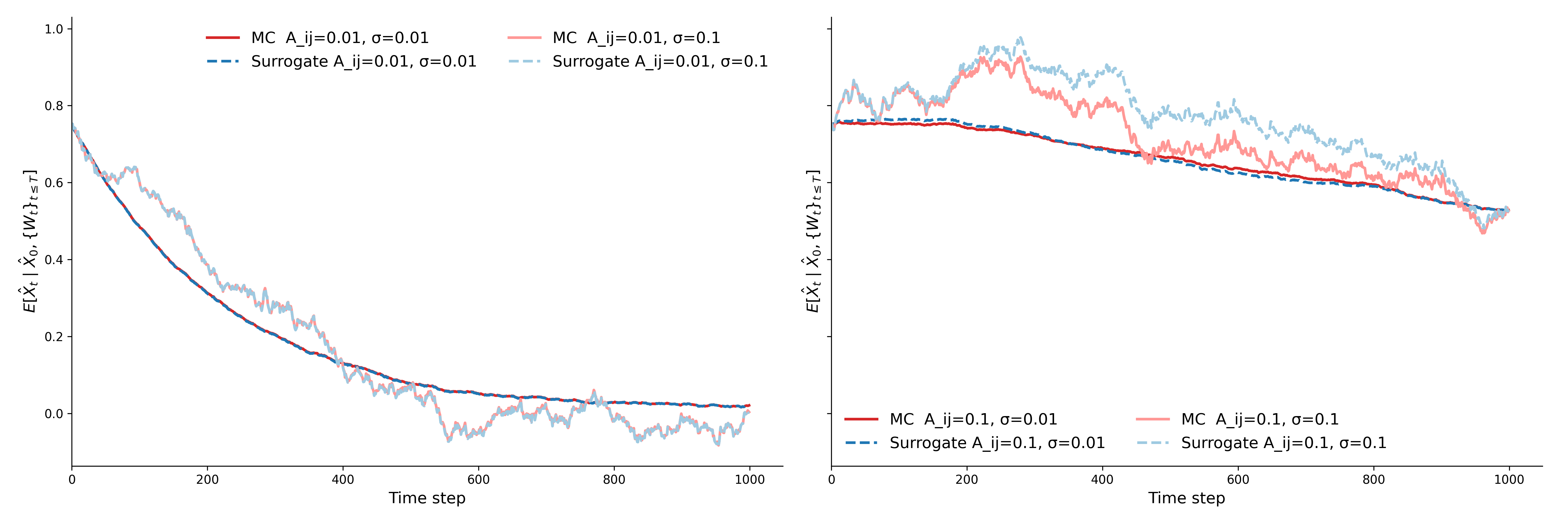}
         \caption{}
    \end{subfigure}
    \begin{subfigure}[b]{0.44\textwidth}
        \includegraphics[width=\textwidth]{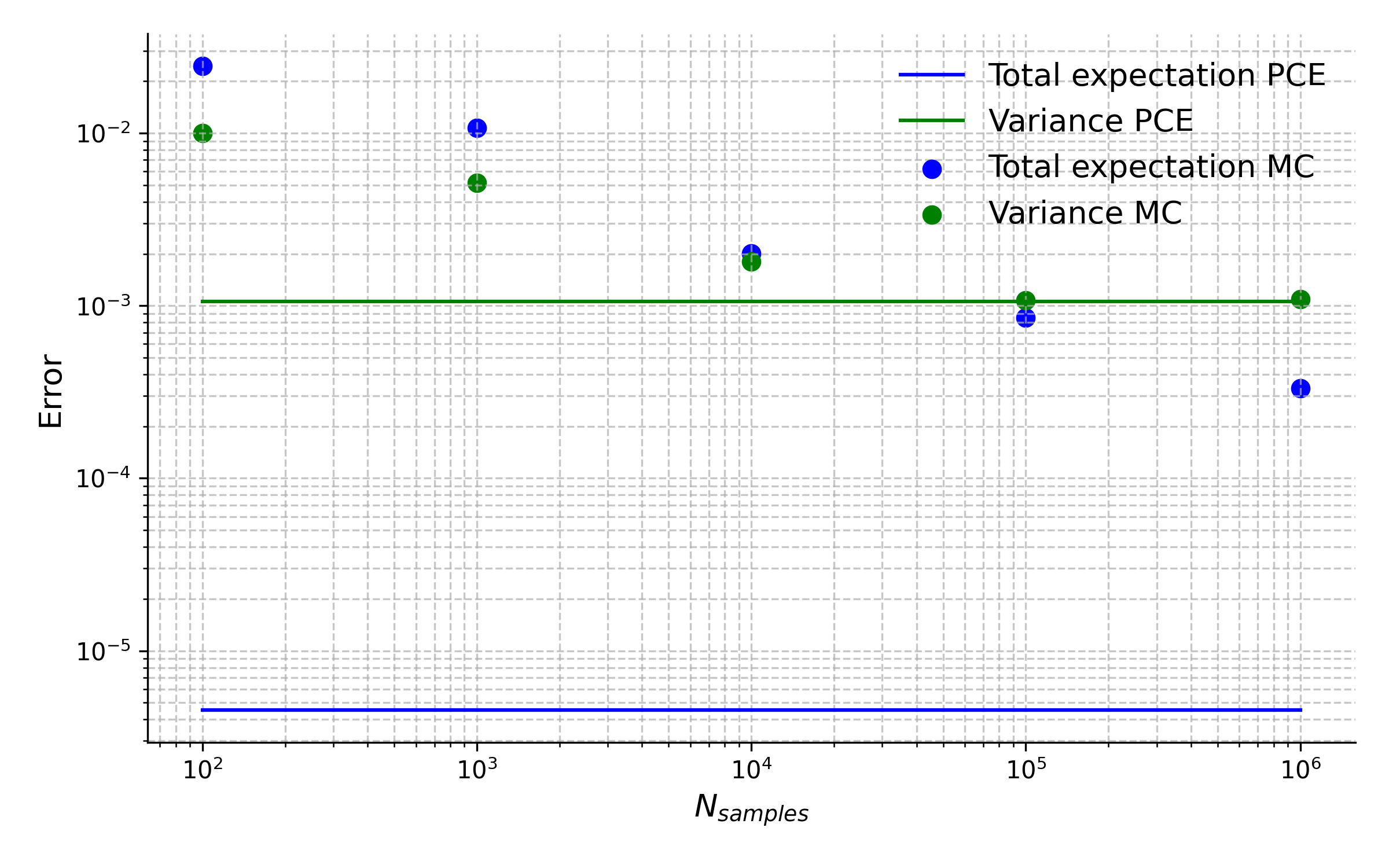}
         \caption{}
    \end{subfigure}
    \hspace{0.01cm}
    \begin{subfigure}[b]{0.44\textwidth}
        \includegraphics[width=\textwidth]{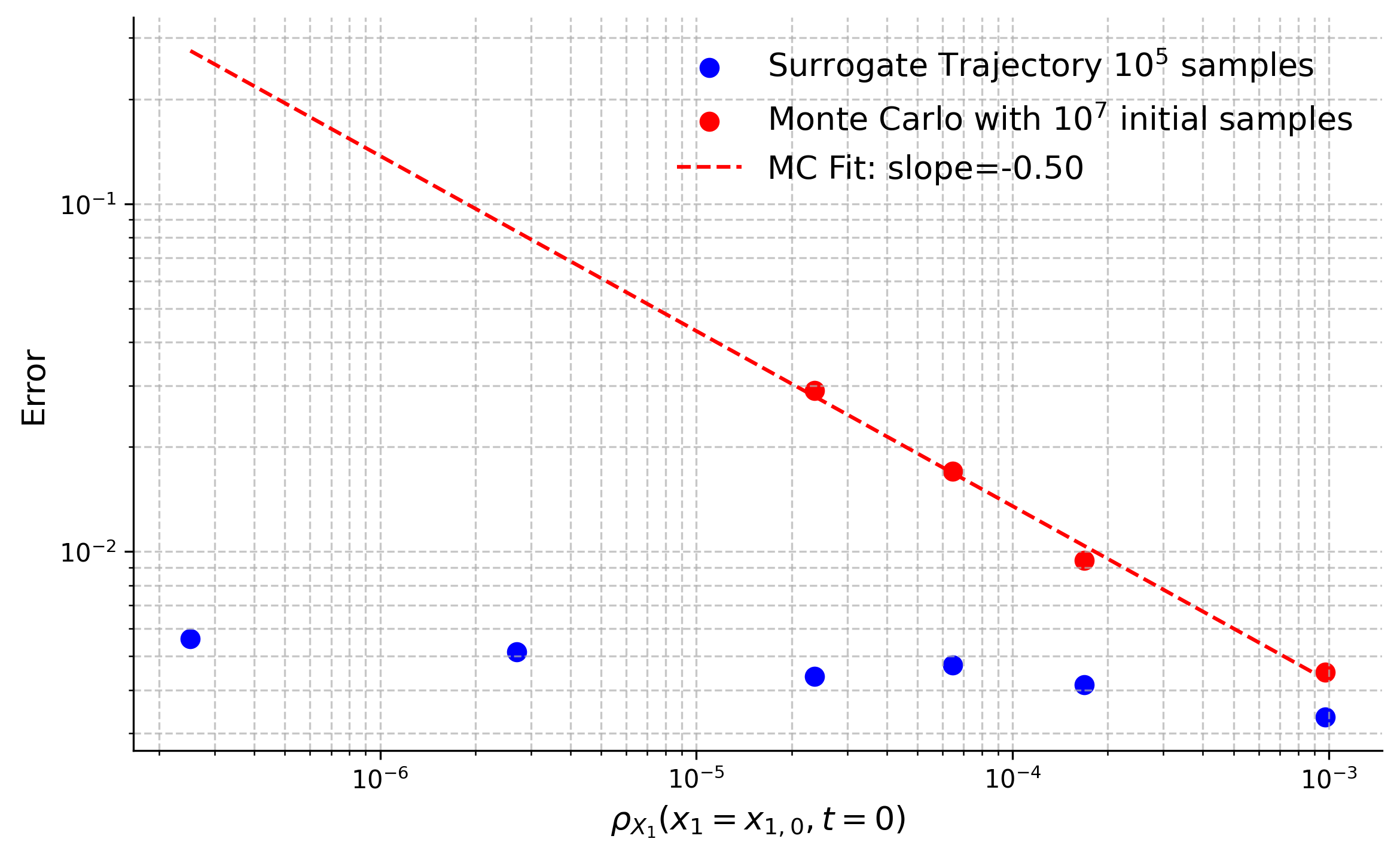}
         \caption{}
    \end{subfigure}
    \begin{subfigure}[b]{0.75\textwidth}
        \includegraphics[width=\textwidth]{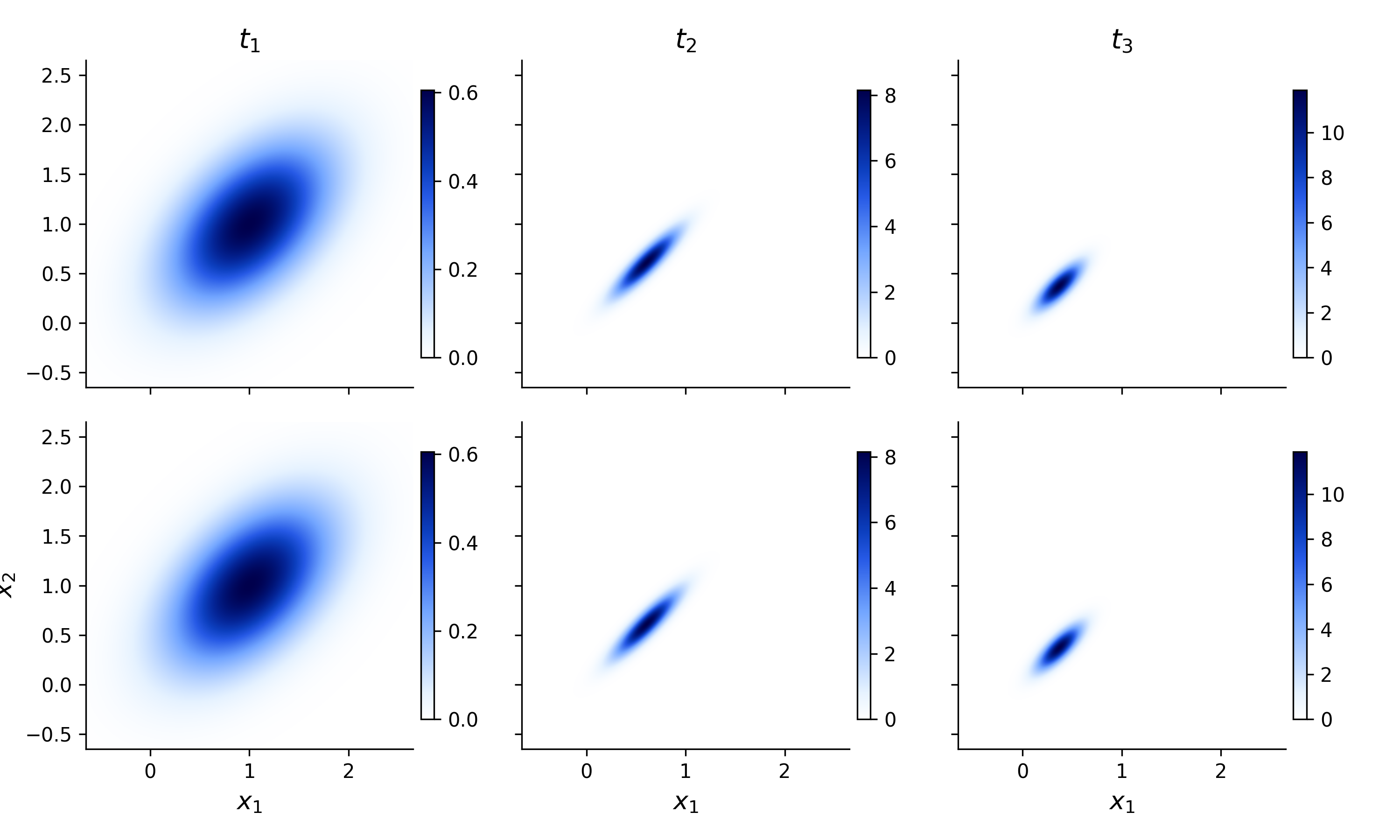}
         \caption{}
    \end{subfigure}
    
    \caption{\justifying 10D linear SDE. (a) Trajectories: weak (left) and stronger (right) interactions settings. (b) Error in total statistics $\mathbb{E}[\hat{X}_t]$ and $\mathrm{Var}(\hat{X}_t)$. (c) Error in $\mathbb{E}[\hat{X}_t\mid\hat{X}_0]$ with respect to the initial probability density of $\hat{X}_0$. (d) Evolution of the marginal joint density $\rho(x_1,x_2)$: PC (top), exact (bottom).
    }
    \label{fig:10D linear sde}
\end{figure*}

\section{Conclusions}\label{sec:conclusions}
\noindent We presented a reduction strategy that generates surrogate trajectories for selected observables while exactly preserving their marginal probability evolution. The approach relies on time-dependent optimal projections and an efficient representation of the probability flow, obtained via PC expansions combined with regression-based conditional expectations. This construction yields an explicit, trajectory-level approximation of high-dimensional dynamics that remains accurate even in low-probability regions, where sampling-based methods rapidly deteriorate. \\ \ \\
At the theoretical level, the reduced dynamics is derived from a pseudo-Markovian decomposition in which the drift and, when relevant, the diffusion coefficients are optimally projected with respect to the instantaneous probability measure. This ensures orthogonality of the unresolved fluctuations and eliminates memory terms. Under general regularity assumptions, the reduced system preserves the marginal law of the full dynamics, both for deterministic flows and Itô diffusions, and provides the best $L^2(\mu_t)$ approximation of the vector field at each time. \\ \ \\
The framework can be viewed as a probability-adapted, time-dependent reformulation of the MZ formalism. By continuously updating the projection according to the evolving law, the method suppresses the memory kernel that otherwise renders MZ reduction intractable. This perspective clarifies the role of conditional expectations in first-order optimal prediction and demonstrates how our construction extends their validity to longer time horizon.
\\ \ \\
Because it yields trajectory-level reduced models that remain statistically consistent with the full system, the method is naturally suited for nonlinear and stochastic dynamics in which only a small set of observables is physically relevant. Potential applications include multiscale systems, molecular and network dynamics, and rare-event regimes where direct Monte-Carlo methods could become ineffective. While the use of polynomial bases introduces practical limitations, these constraints can be relaxed by adopting alternative bases or data-driven approximations of the probability flow. Overall, the framework offers a principled reduction strategy grounded in statistical physics, with the potential to support analysis and computation in a broad class of complex dynamical systems.
\section*{Acknowledgments}
    \noindent The authors would like to acknowledge Prof. Fabio Nobile and Ramzi Dakhmouche for the fruitful discussions on probability flows and measure theoretic aspects.
    This project is supported by the Swiss National Science Foundation grant 212876.
\section*{Author Declarations}
\noindent The authors have no conflicts to disclose.

\section*{Data Availability Statement}
\noindent The data that support the findings of this study are openly available in Zenodo at \url{https://doi.org/10.5281/zenodo.17601996}, reference number 10.5281/zenodo.17601996.

\appendix

\section{Theoretical results}\label{app:proofs}
\noindent The working assumptions are listed in Section~\ref{sec:assumptions}.

\subsection{Conditional expectation}\label{app:condexp}
\begin{proof}[Proof of Proposition~\ref{prop:lipschitz mckean-vlasov} (Lipschitzness of Conditional Expectation)]
    Assume $f:\Gamma\to\mathbb{R}$ is globally Lipschitz. Then, for any $y,z\in\hat{\Gamma}$, we have,
    \begin{eqnarray}
        \big| \mathbb{E}_{\mu_t}[f(X)\mid \hat{X}=y] -  \mathbb{E}_{\mu_t}[f(X)\mid \hat{X}=z]\big| &=& \big| \mathbb{E}_{\mu_{t|\hat{X}=y}}[f(X)] -  \mathbb{E}_{\mu_{t|\hat{X}=z}}[f(X)]\big| \nonumber\\
        &\leq& C_{1,t}\mathcal{W}_1(\mu_{t|\hat{X}=y},\mu_{t|\hat{X}=z})\nonumber\\
        &\leq &C_t|y-z|,
    \end{eqnarray}
    where $C_{1,t},C_t>0, \forall t\geq 0$. The first inequality follows from Kantorovich-Rubinstein duality theorem \cite{Villani2009}.
\end{proof}

\subsection{Pseudo-Markovian dynamics}
\subsubsection{Exact trajectory}
\begin{proof}[Proof of Theorem~\ref{thm:pseudo-mark} (Trajectory Consistency for ODE)]\label{proof:pseudomark}
Fix $\zeta^n$ as the orthogonal complement to $\mathbb{E}_{\mu_t}[b(X)\mid \hat{X}^n]$, i.e. $\zeta^n:= \hat{b}(\hat{X}^n) - \mathbb{E}_{\mu_t}[\hat{b}(X)\mid \hat{X}^n]$, where the resolved variables $\hat{X}^n$ follow
the Pseudo-Markovian scheme,
    \begin{eqnarray}
        d\hat{X}^n=\mathbb{E}_{\mu_t}[\hat{b}(X)\mid \hat{X}_n]\Delta t + \zeta^n\Delta t, \quad \hat{X}^0 = X_{i,0}. 
    \end{eqnarray}
\begin{enumerate}
    \item Let $X(t_n)$ be the solution of
    \begin{eqnarray}
        \frac{d}{dt}X_t=b(X_t), \quad X_0\sim\mu_0(X)\in\mathcal{P}_2,
    \end{eqnarray}
    at time $t=t_n$ and let $X_i(t_n)$ be the corresponding solution for the resolved variable. Consider its approximation $X_i^n$ resulting from the first-order Euler scheme,
    \begin{eqnarray}
        dX^n=b(X^n)\Delta t, \quad X^0=X_0.
    \end{eqnarray}
    Then, from the definition of $\zeta^n$ and the Lipschitzness of $b$, we have,
    \begin{eqnarray}
        \mathbb{E}[|\hat{X}^{n+1}-X_i^{n+1}|^2] &\leq& 3\mathbb{E}[|\hat{X}^{n}-X_i^{n}|^2] + 3\mathbb{E}[|\hat{b}(\hat{X}^n)-\hat{b}(X_i^n)|^2] \Delta t^2 \nonumber\\
        &\leq& (3 + C\Delta t^2)\mathbb{E}[|\hat{X}^{n}-X_i^{n}|^2] \nonumber\\
        &\leq&  (3 + C\Delta t^2)^{n+1} \mathbb{E}[|\hat{X}^{0}-X_i^{0}|^2] = 0,
    \end{eqnarray}
    which implies $\hat{X}^n=X^n_i \quad a.s. \quad \forall n$. Then,
    \begin{eqnarray}
        \mathbb{E}[|\hat{X}^{n}-X_i(t_n)|^2] = \mathbb{E}[|{X}_i^{n}-X_i(t_n)|^2] = \mathbb{E}[|{\phi}_i^{n}(X_0)-\phi_i(t_n,X_0)|^2]
    \end{eqnarray}
    where we defined the (flow) maps $\phi^n(X_0):=X^n \ s.t. \ X^0=X_0$ and $\phi(t_n,X_0):=X(t_n) \ s.t. \ X(t_0)=X_0$. 
    Finally, since the first-order Euler scheme converges with rate $O(\Delta t)$ \cite{Süli_Mayers_2003}, we have,
    \begin{eqnarray}
        |\phi^n(X_0)-\phi(t_n,X_0)|^2 = O(\Delta t^2),
    \end{eqnarray}
    for any initial condition $X_0$. Noting that taking expectation preserves the rate, we have,
    \begin{eqnarray}
        \mathbb{E}[|\hat{X}^n-X_i(t_n)|^2]=\mathbb{E}[|X^n-X(t_n)|^2] = \mathbb{E}_{\mu_0}[|{\phi}_i^{n}(X_0)-\phi_i(t_n,X_0)|^2] = O(\Delta t^2).
    \end{eqnarray}
    \item The first equality is proved using the law of total expectation. The second equality follows from the definition of $\zeta^n$ and the last equality is shown again using the law of total expectation.
\end{enumerate}
\end{proof}

\begin{proof}[Proof of Corollary~\ref{corol:pseudo-mark-sde} (Trajectory Consistency for SDE)]\label{proof:pseudomark sde}
Fix $\zeta^n$ as the orthogonal complement to $\mathbb{E}_{\mu_t}[b(X)\mid \hat{X}^n]$, i.e. $\zeta^n:= \hat{b}(\hat{X}^n) - \mathbb{E}_{\mu_t}[\hat{b}(X)\mid \hat{X}^n]$, and $\eta^n$ as the orthogonal complement to $\mathbb{E}_{\mu_t}[\hat{\sigma}(X)\mid \hat{X}^n]\xi^n$, i.e. $\eta^n:= \hat{\sigma}(\hat{X}^n)\xi^n - \mathbb{E}_{\mu_t}[\hat{\sigma}(X)\mid \hat{X}^n]\xi^n$, 
where $\xi^n$ is a standard normal random variable and the resolved variables $\hat{X^n}$ follow
the Pseudo-Markovian scheme,
    \begin{eqnarray}
        d\hat{X}^n=\mathbb{E}_{\mu_t}[\hat{b}(X)\mid \hat{X}_n]\Delta t + \mathbb{E}_{\mu_t}[\hat{\sigma}(X)\mid \hat{X}_n]\xi^n\sqrt{\Delta t} +(\zeta^n\Delta t + \eta^n\sqrt{\Delta t}), \quad \hat{X}^0 = X_{i,0}. 
    \end{eqnarray}
\begin{enumerate}
    \item Let $X(t_n)=:F_0^{t_n}(\int_0^{t_n}dW_t,X_0)$ be the solution of
    \begin{eqnarray}
        dX_t=b(X_t)dt+\sigma(X_t) dW_t, \quad X_0\sim\mu_0(X)\in\mathcal{P}_2,
    \end{eqnarray}
    at time $t=t_n$ with fixed realization of the Wiener process $\{W_t\}_{t\leq t_n}$ and let $X_i(t_n)=:F_{i,0}^{t_n}(\int_0^{t_n}d W_t,X_0)$ be the corresponding solution for the resolved variable. Consider its approximation $X_i^n=:F_{i}^n(\{\xi_k\}_{k\leq n},X_0)$ resulting from the first-order Euler-Maruyama scheme,
    \begin{eqnarray}
        dX^n=b(X^n)\Delta t + \sigma(X^n)\xi^n\sqrt{\Delta t}, \quad X^0=X_0.
    \end{eqnarray}
    Let $\hat{X}^n=:\hat{F}^n(n,\{\xi_k\}_{k\leq n},\hat{X_0})$ and fix the realization of the random process $\xi^n$.
    Then, from the definition of $\zeta^n$ and $\eta^n$, and the Lipschitzness of $b$ and $\sigma$, we have,
    \begin{eqnarray}
        \mathbb{E}[|\hat{X}^{n+1}-X_i^{n+1}|^2] &\leq& 3\mathbb{E}[|\hat{X}^{n}-X_i^{n}|^2] + 3\mathbb{E}[|\hat{b}(\hat{X}^n)-\hat{b}(X_i^n)|^2] \Delta t^2 +3\mathbb{E}_{\mu_t}[|(\hat{\sigma}(\hat{X}^n)-\hat{\sigma}(X_i^n)){\xi^n}|^2]\Delta t\nonumber\\
        &\leq&3\mathbb{E}[|\hat{X}^{n}-X_i^{n}|^2] + 3\mathbb{E}[|\hat{b}(\hat{X}^n)-\hat{b}(X_i^n)|^2] \Delta t^2 +3\mathbb{E}_{\mu_t}[|\hat{\sigma}(\hat{X}^n)-\hat{\sigma}(X_i^n)|^2]\Delta t\nonumber\\
        &\leq& (3 + K\Delta t+ C\Delta t^2)\mathbb{E}[|\hat{X}^{n}-X_i^{n}|^2] \nonumber\\
        &\leq&  (3 + K\Delta t + C\Delta t^2)^{n+1} \mathbb{E}[|\hat{X}^{0}-X_i^{0}|^2] = 0,
    \end{eqnarray}
    which implies $\hat{X}^n=X^n_i \quad a.s. \quad \forall n$. Then,
    \begin{eqnarray}
        \mathbb{E}[|\hat{X}^{n}-X_i(t_n)|^2] = \mathbb{E}[|{X}_i^{n}-X_i(t_n)|^2] = \mathbb{E}[|\hat{F}_i^{n}(\{\xi^k\}_{k\leq n},X_0)-F_0^{t_n}(\int_0^{t_n}dW_t,X_0)|^2].
    \end{eqnarray}
    Finally, since the first-order Euler-Maruyama scheme converges with rate $O({\Delta t})$ in $L^2$, we have,
    \begin{eqnarray}
        \mathbb{E}[|\hat{X}^n-X_i(t_n)|^2]&=&\mathbb{E}[|X^n-X(t_n)|^2] \nonumber\\
        &=& \mathbb{E}_{\mu_0}[|\hat{F}_i^{n}(\{\xi^k\}_{k\leq n},X_0)-F_0^{t_n}(\int_0^{t_n}dW_t,X_0)|^2]\nonumber \\
        &=& O(\Delta t^2).
    \end{eqnarray}
    \item The first equality is proved using the law of total expectation. The second equality follows from the definition of $\zeta^n$ and $\eta^n$, and the last equality is shown again using the law of total expectation.
\end{enumerate}

\end{proof}
\subsubsection{Surrogate trajectory}
\noindent\textbf{Deterministic dynamics}\label{app:proof x ode}

\begin{proof}[Proof of Proposition~\ref{prop:deterministic} (Marginal Preservation: Deterministic Case)]
    Let the full-order state vector $X\in\Gamma \subseteq \mathbb{R}^N$. Consider the dynamical system 
    \begin{eqnarray}
        \frac{dX}{dt} = b(X), \quad X(t=0) = X_0 \sim \mu_0,
    \end{eqnarray}
    where $b:\mathbb{R}^N\to\mathbb{R}^N$. Let $\hat{X}=(X_1,\dots,X_m), \ m<N$.
    Then, for any $g(\hat{X})\in C_0^\infty(\hat{\Gamma})$, we have,
    \begin{eqnarray}
        \frac{d}{dt}g = \sum_{i=1}^N b_i\partial_{x_i}g.
    \end{eqnarray}
    Therefore,
    \begin{eqnarray}
        \int_\Gamma \frac{d}{dt}g \ d\mu = \int_\Gamma \sum_{i=1}^N b_i        \frac{d}{dt}g = \sum_{i=1}^N b_i\partial_{x_i}g \ d\mu.
    \end{eqnarray}
    Since it is assumed that $\mu(X)$ admits a density $\rho(x)$, we get,
    \begin{eqnarray}
        \int_\Gamma g(\hat{x}) \ \partial_t\rho \ dx =  \int_{\hat{\Gamma}} g(\hat{x}) \ \partial_t\hat{\rho} \ d\hat{x} &=& \int_\Gamma \sum_{i=1}^N b_i\partial_{x_i}g(\hat{x}) \rho \ dx\nonumber\\
        &=&\int_{\hat{\Gamma}} \sum_{i=1}^m\mathbb{E}_{\mu_t}[b_i\mid \hat{x}] (\partial_{x_i}g) \hat{\rho}\ d\hat{x}\nonumber\\
        &=& -\int_{\hat{\Gamma}} g \sum_{i=1}^m\partial_{x_i}(\mathbb{E}_{\mu_t}[b_i\mid \hat{x}] \hat{\rho})\ d\hat{x},
    \end{eqnarray}
    where $\hat{\rho}(\hat{x}):=\int_{\hat{\Gamma}^\perp}\rho\ d\hat{x}^\perp$. This relation holds for any $g(\hat{X})\in C_0^\infty(\hat{\Gamma})$. Therefore, $\hat{\rho}$ is a weak solution of, 
    \begin{eqnarray}\label{eq:proba flow reference}
        \partial_tp+ \sum_{i=1}^m\partial_{x_i}(\mathbb{E}_{\mu_t}[b_i\mid \hat{x}] p)=0.
    \end{eqnarray}
    Similarly, suppose $Y\in\hat{\Gamma}\subseteq\mathbb{R}^m$ satisfies the reduced dynamics,
    \begin{eqnarray}
        \frac{d}{dt}Y=\mathbb{E}_{\mu_t}[\hat{b}\mid  Y], \quad Y_0\sim \hat{\nu}_0 = \hat{\mu}_0
    \end{eqnarray}
    where $\hat{b}:=(b_1,\dots,b_m)$ and the law of $Y$ is $\hat{\nu}$, which is assumed to admit a density $f(y)$. Then, following the same reasoning, for any $h(Y)\in C_0^\infty(\hat{\Gamma})$, we have,
    \begin{eqnarray}
        \frac{d}{dt}h = \sum_{i=1}^m \mathbb{E}_{\mu_t}[b_i\mid Y]\partial_{x_i}h,
    \end{eqnarray}
    and,
    \begin{eqnarray}
         \int_{\hat{\Gamma}} h \ \partial_tf \ dy = -\int_{\hat{\Gamma}} h \sum_{i=1}^m\partial_{x_i}(\mathbb{E}_{\mu_t}[b_i\mid \hat{y}] f)\ dy.
    \end{eqnarray}
    Hence, $f$ is also a weak solution of Eq.~(\ref{eq:proba flow reference}).
    Therefore, $\hat{X}$ and $Y$ are equal in law.
    Finally, by the orthogonal projection theorem in Hilbert spaces, the time-dependent conditional expectation $\mathbb{E}_{\mu_t}[(\cdot)\mid\hat{X}=\hat{X}_t]$ provides the best $L^2$-approximation in $\hat{\mathcal{H}}_t$ of any function in $\mathcal{H}_t$, $\forall t \geq 0$. Hence, the vector field $b$ is optimally projected in a time-dependent fashion.
\end{proof}

\noindent\textbf{Stochastic dynamics}\label{app:surrogate sde}
\begin{proof}[Proof of Proposition~\ref{prop:stochastic} (Marginal Preservation: Stochastic Case)]
    Let the full-order state vector $X\in\Gamma \subseteq \mathbb{R}^N$. Consider the Itô SDE, 
    \begin{eqnarray}
        dX_t = b(X_t)dt+\sigma(X_t)dW_t, \quad X(t=0) = X_0 \sim \mu_0,
    \end{eqnarray}
    where $b:\mathbb{R}^N\to\mathbb{R}^N$ and $\sigma:\Gamma\to\mathbb{R}^{N\times N}$. Let $\hat{X}=(X_1,\dots,X_m), \ m<N$.
    Then, for any $g(\hat{X})\in C_0^\infty(\hat{\Gamma})$, we have by Itô's lemma,
    \begin{eqnarray}
        dg(X_t) = Lg(X_t)dt + \sum_{i,j=1}^N\partial_{x_i}g(\hat{X}_t)\sigma_{ij}({X}_t)dW_{j,t},
    \end{eqnarray}
    where $ L=\sum_i b_i \,\partial_{x_i} + \sum_{i,j} D_{ij}\,\partial_{x_i}\partial_{x_j}$ is the generator is of the stochastic process and $D=\tfrac12 \sigma\sigma^\top$. Taking the expectation, we obtain,
    \begin{eqnarray}
        \frac{d}{dt}\mathbb{E}[g(X_t)]=\mathbb{E}[Lg(X_t)].
    \end{eqnarray}
    Assuming the law of $X_t$ admits a density $\rho(x)$, we can write,
    \begin{eqnarray}
        \int_{\hat{\Gamma}} g(\hat{x}) \ \partial_t\hat{\rho} \ d\hat{x} 
        &=&\int_\Gamma g(\hat{x}) \ \partial_t\rho \ dx  \nonumber\\
        &=& \int_\Gamma (Lg(\hat{x})) \rho \ dx\nonumber\\
        &=& \int_\Gamma (\sum_{i=1}^N b_i\partial_{x_i}g(\hat{x}) +\sum_{i,j=1}^N D_{ij}\partial_{x_i}\partial_{x_j} g(\hat{x}))\rho \ dx \nonumber\\
        &=&\int_{\hat{\Gamma}} \left(\sum_{i=1}^m\mathbb{E}_{\mu_t}[b_i\mid \hat{x}] \partial_{x_i}g \ + \ \sum_{i,j=1}^m \mathbb{E}_{\mu_t}[D_{ij}\mid \hat{x}]\partial_{x_i}\partial_{x_j} g \right)\hat{\rho} \ d\hat{x} \nonumber\\
        &=& \int_{\hat{\Gamma}} g \left(-\sum_{i=1}^m\partial_{x_i}(\mathbb{E}_{\mu_t}[b_i\mid \hat{x}]\hat{\rho})\ + \sum_{i,j=1}^m \partial_{x_i}\partial_{x_j}(\mathbb{E}_{\mu_t}[D_{ij}\mid \hat{x}] \hat{\rho}) \right)\ d\hat{x},
    \end{eqnarray}
    where $\hat{\rho}(\hat{x}):=\int_{\hat{\Gamma}^\perp}\rho\ d\hat{x}^\perp$. This relation holds for any $g(\hat{X})\in C_0^\infty(\hat{\Gamma})$. Therefore, $\hat{\rho}$ is a weak solution of, 
    \begin{eqnarray}\label{eq:proba flow sde reference}
        \partial_tp=- \sum_{i=1}^m\partial_{x_i}(\mathbb{E}_{\mu_t}[b_i\mid \hat{x}] p)+\sum_{i,j=1}^m \partial_{x_i}\partial_{x_j}(\mathbb{E}_{\mu_t}[D_{ij}\mid \hat{x}]p).
    \end{eqnarray}
    Similarly, suppose $Y\in\hat{\Gamma}\subseteq\mathbb{R}^m$ satisfies the reduced stochastic dynamics,
    \begin{eqnarray}
        dY_t=\mathbb{E}_{\mu_t}[\hat{b}\mid  Y_t]dt + \mathbb{E}_{\mu_t}[\hat{\sigma}\mid Y_t]d\hat{W}_t, \quad Y_0\sim \hat{\nu}_0 = \hat{\mu}_0,
    \end{eqnarray}
    where $\hat{b}:=(b_1,\dots,b_m)$, $\hat{\sigma}:=\sigma_{ij}, \ i,j=\{1,\dots,m\}$ and the law of $Y_t$ is $\hat{\nu}_t$, which is assumed to admit a density $f(y)$. Then, following the same reasoning, for any $h(Y)\in C_0^\infty(\hat{\Gamma})$, we have,
    \begin{eqnarray}
        \frac{d}{dt}\mathbb{E}[h] = \mathbb{E}\bigg[\sum_{i=1}^m \mathbb{E}_{\mu_t}[b_i\mid Y]\partial_{x_i}h + \sum_{i,j=1}^m\mathbb{E}_{\mu_t}[D_{ij}\mid Y]\partial_{x_i}\partial_{x_j} h  \bigg],
    \end{eqnarray}
    and,
    \begin{eqnarray}
         \int_{\hat{\Gamma}} h \ \partial_tf \ dy = \int_{\hat{\Gamma}} h \left(-\sum_{i=1}^m\partial_{x_i}(\mathbb{E}_{\mu_t}[b_i\mid \hat{y}] f) + \sum_{i,j=1}^m\partial_{x_i}\partial_{x_j}(\mathbb{E}_{\mu_t}[D_{ij}\mid Y] f) \right)\ dy.
    \end{eqnarray}
    Hence, $f$ is also a weak solution of Eq.~(\ref{eq:proba flow sde reference}).
    Therefore, $\hat{X}$ and $Y$ are equal in law.
    Finally, by the orthogonal projection theorem in Hilbert spaces, the time-dependent conditional expectation $\mathbb{E}_{\mu_t}[(\cdot)\mid\hat{X}=\hat{X}_t]$ provides the best $L^2$-approximation in $\hat{\mathcal{H}}_t$ of any function in $\mathcal{H}_t$, $\forall t \geq 0$. Hence, the generator is optimally projected in a time-dependent fashion.
\end{proof}
\section{Error estimate}\label{app:error}
\subsection{Short-time limit}
\noindent To derive the short-time limit of Eq.~(\ref{eq:error analytical}), we write its second-order Taylor expansion about $t=0$. We first consider, 
\begin{eqnarray}
    \overline{F}(t) = \overline{F}(0) + t\overline{F}'(0) + O(t^2),\\
    Q(t) = Q(0) + tQ'(0) + O(t^2).
\end{eqnarray}
Since the propagator \(\Phi(t, \tau)\) satisfies the differential equation,
\[
\frac{d}{dt} \Phi(t, \tau) = Q(t) \Phi(t, \tau), \quad \Phi(\tau, \tau) = I,
\]
we can expand \(\Phi(t, \tau)\) for \(t\) close to \(\tau = 0\). We write the time-ordered exponential as,
\[
\Phi(t, \tau) = I + \int_{\tau}^{t} Q(s) \, ds + \frac{1}{2} \int_{\tau}^{t} \int_{\tau}^{s_1} Q(s_1) Q(s_2) \, ds_2 \, ds_1 + O(t^3).
\]
Using,
\[
Q(s) = Q(0) + s Q'(0) + \cdots,
\]
we obtain,
\[
\int_{\tau}^{t} Q(s) \, ds = Q(0)(t - \tau) + \frac{1}{2} Q'(0)(t^2 - \tau^2) + O(t^3, \tau^3)
\]
\[
\frac{1}{2} \int_{\tau}^{t} \int_{\tau}^{s_1} Q(s_1) Q(s_2) \, ds_2 \, ds_1 = \frac{1}{2} Q(0)^2(t - \tau)^2 + O((t-\tau)^3).
\]
Therefore, up to second order, we have,
\[
\Phi(t, \tau) \approx I + (t - \tau) Q(0) + \frac{1}{2} (t - \tau)^2 Q(0)^2 + \frac{1}{2}Q'(0)(t^2-\tau^2) + O(t^3).
\]
We can now use this expression along with the Taylor expansion of $\overline{F}(t)$ in Eq.~(\ref{eq:error analytical}). We get the first and second order contributions,
\[
-\int_0^t I \overline{F}(0) = -t\overline{F}(0),
\]
\[
-\int_0^t (t - \tau) Q(0) \overline{F}(0) \, d\tau = -Q(0) \overline{F}(0) \int_0^t (t - \tau) d\tau = -Q(0) \overline{F}(0) \frac{t^2}{2},
\]
\[
-\int_0^t \tau \, \overline{F}'(0) \, d\tau = -\frac{t^2}{2} \, \overline{F}'(0),
\]
We get the final expression,
\begin{eqnarray}
    e(t) = -t\overline{F}(0) - \frac{t^2}{2}[Q(0)\overline{F}(0)+\overline{F}'(0)] + O(t^3).
\end{eqnarray}
However, from the definition of $\overline{F}(t)$, we get $\overline{F}(0)=0$, and hence,
\begin{eqnarray}
    e(t) =  -\frac{t^2}{2}\overline{F}'(0) + O(t^3).
\end{eqnarray}
In the special case where the drift is linear and the initial probability measure Gaussian, we can derive an explicit expression for the time derivative $\overline{F}'(0)$. By definition, we have
\[
F(t)=B_{res}(\mathbb{E}[ \ X_t \ | \ \hat{X}_0 \ ] - \mathbb{E}[ \ \mathbb{E}[\ X_t \ | \ \hat{X}_t\ ]\ | \hat{X}_0 \ ]),
\]
where $B_{res} = [B_{res,res}, B_{res,unres}]$ is the matrix formed out of the first $m$ rows of $B$, where we recall that $m$ is the number of resolved variables. We have,
\[
\frac{d}{dt}\mathbb{E}[ \ X_t \ | \ \hat{X}_0 \ ] = B\mathbb{E}[ \ X_t \ | \ \hat{X}_0 \ ],
\]
where we used $X_t = e^{Bt}X_0$. Furthermore,
\[
\frac{d}{dt} \mathbb{E}\left[ \mathbb{E}[X_t \mid \hat{X}_t] \,\big|\, \hat{X}_0 \right]
= \mathbb{E}\left[ \frac{d}{dt} \mathbb{E}[X_t \mid \hat{X}_t] \,\big|\, \hat{X}_0 \right].
\]
Using,
\[
\mathbb{E}[X_t \mid \hat{X}_t] = \mu_t + C(t)(\hat{X}_t - \hat{\mu}_t),
\]
\[
C(t) = \Sigma_{all, res}(t)\Sigma^{-1}_{res,res}(t) = Cov(X_t,\hat{X}_t)Cov(\hat{X}_t)^{-1},
\]
where the subscript $(\cdot)_{all,res}$ indicates the sub-matrix obtained by taking all the rows but only the columns corresponding to the resolved components of the original matrix, we obtain,
\[
\frac{d}{dt} \mathbb{E}[X_t \mid \hat{X}_t] = \dot{\mu}_t + \dot{C}(t)(\hat{X}_t - \hat{\mu}_t) + C(t) (B_{res}\hat{X}_t-\dot{\hat{\mu}}_t).
\]
Let us then compute,
\[
\overline{F}'(0) = B_{\text{res}} \left( \left. \frac{d}{dt} \mathbb{E}[X_t \mid \hat{X}_0] \right|_{t=0} 
- \left. \mathbb{E}\left[ \frac{d}{dt} \mathbb{E}[X_t \mid \hat{X}_t] \,\big|\, \hat{X}_0 \right] \right|_{t=0} \right).
\]
At $t=0$, we have,
\[
\mu_0 = m_0, \quad \dot{\mu}_0 = B m_0, \quad \hat{X}_0 \text{ fixed}.
\]
The second term becomes,
\[
\mathbb{E}\left[ \dot{\mu}_0 + \dot{C}(0)(\hat{X}_0 - \hat{\mu}_0) + C(0)(B_{res}\hat{X_0} -\dot{\hat{\mu}}_0) \,\big|\, \hat{X}_0 \right]
= B m_0 + \dot{C}(0)(\hat{X}_0 - \hat{m}_0) + C(0)B_{res}(\hat{X}_0 - \hat{m}_0).
\]
Putting everything together,
\begin{eqnarray}
\overline{F}'(0) &=& B_{\text{res}} \left( Bm_0 + C(0)B_{res}(\hat{X_0}-\hat{m}_0) - \left( B m_0 + \dot{C}(0)(\hat{X}_0 - \hat{m}_0) + C(0) B_{res} (\hat{X}_0 - \hat{m}_0) \right) \right)\nonumber\\
 &=& -B_{res} \dot{C}(0)(\hat{X}_0 - \hat{m}_0).   
\end{eqnarray}
\subsection{Stochastic dynamics}\label{app:error sde}
\noindent We generalize the analysis of the modeling error, arising from the approximation of the evolution of the quantity $\mathbb{E}[\hat{X}_t\mid \hat{X}_0]$ with our pseudo-Markovian reduced dynamics Eq.~\eqref{eq:optimal prob consistent sde} , when the latter is built on the stochastic process, Eq.~(\ref{eq:sde assumption}), at the micro-scale. We derive the result for a constant noise coefficient matrix $\sigma\in \mathbb{R}^{N\times q}$. In this case $\hat{\sigma}=\sigma_{res}$, where $\hat{\sigma}:\frac{1}{2}\sigma\sigma^T=D_{res}$ and $\sigma_{res}$ and $D_{res}$ are respectively the sub-matrices of $\sigma$ and $D$ formed by selecting the rows and columns corresponding to the resolved components. 
To compute the modeling error, we start by writing the exact evolution of the resolved components $\hat{X}=(X_1,\dots,X_m)\in\mathbb{R}^m$,
\begin{eqnarray}\label{eq:exact resolved sde}
    \frac{d}{dt}\hat{X}^{exact}_t = \mathbb{E}_{\mu_t}[\hat{b}(X^{exact})\mid\hat{X}^{exact}_t] + \hat{\sigma}dW_t  + F_t,
\end{eqnarray}
with the initial condition $\hat{X}^{exact}_0 = \hat{X}_0$ and
where the $\mathbb{E}_{\mu_t}[\:(\cdot)\mid\hat{X}_t\:]$-orthogonal random force $F_t \in \mathbb{R}^m$ is simply the orthogonal complement of our optimal projection of the resolved components $\hat{b}\in\mathbb{R}^m$ of the drift vector $b \in \mathbb{R}^N$. 
The unresolved variables $\tilde{X}_0^{exact} \in\mathbb{R}^{N-m}$ are sampled from $\mu_{0|\hat{X}_0}$, the initial probability measure conditioned on $\hat{X}_0$.  
A surrogate trajectory $\hat{X}^S_t$ is a solution of Eq.~(\ref{eq:optimal prob consistent sde}), with initial condition $\hat{X}^{S}_0 = \hat{X}_0$.
It is an approximation of the quantity $\mathbb{E}[\hat{X_t}\mid\hat{X}_0]$, whose exact evolution is derived by taking the conditional expectation of Eq.~(\ref{eq:exact resolved sde}), given $\hat{X}_0$,
\begin{eqnarray}
    \frac{d}{dt}\mathbb{E}[\hat{X}^{exact}_t\mid\hat{X}_0] = \mathbb{E}\left[\mathbb{E}_{\mu_t}[\hat{b}(X^{exact})\mid\hat{X}^{exact}_t] \:\bigg|\: \hat{X}_0 \right] + \mathbb{E}[F_t\mid\hat{X}_0],\\
    \mathbb{E}[\hat{X}_0^{exact}\mid\hat{X}_0] = \hat{X}_0.
\end{eqnarray}
Defining the error as $e(t)=\hat{X}^S_t-\mathbb{E}[\hat{X}^{exact}_t\mid\hat{X}_0]$, its evolution is governed by,
\begin{eqnarray}
    \frac{d}{dt}e = \mathbb{E}_{\mu_t}[\hat{b}(X^S)\mid\hat{X}^S_t] - \mathbb{E}\left[\mathbb{E}_{\mu_t}[\hat{b}(X^{exact})\mid\hat{X}^{exact}_t] \:\bigg|\: \hat{X}_0 \right] - \mathbb{E}[F_t\mid\hat{X}_0],\\
    e(0)=0.
\end{eqnarray}
For a linear system governed by $\frac{d}{dt}X=BX+\sigma dW_t$, $B \in\mathbb{R}^{N\times N}$, with initial Gaussian probability measure $\mathcal{N}_X(m_0,\Sigma_0)$, we obtain,
\begin{align}
    \frac{\mathrm{d}}{\mathrm{d}t} e(t) 
    &= Q(t)\, e(t) - \overline{F}(t), \label{eq:evolution_error sde} \\
    \overline{F}(t) 
    &:= \mathbb{E} \left[ F_t \,\mid \hat{X}_0 \right], \label{eq:mean_force_stoch} \\
    Q(t) 
    &:= B_{\mathrm{res},\,\mathrm{res}} 
    + B_{\mathrm{res},\,\mathrm{unres}} 
    \Sigma(t)_{\mathrm{unres},\,\mathrm{res}} 
    \Sigma(t)^{-1}_{\mathrm{res},\,\mathrm{res}}, \label{eq:Q_def_sde}
\end{align}
where, $m_0 \in \mathbb{R}^N$ and $\Sigma_0 \in \mathbb{R}^{N\times N}$ are the initial mean and positive semi-definite covariance matrix, respectively.
We defined $B_{\mathrm{res},\,\mathrm{res}} \in \mathbb{R}^{m \times m}$ as the $m\times m$ sub-matrix of $B$, obtained by selecting the rows and columns corresponding to the resolved variables. Similarly, $B_{\mathrm{res},\,\mathrm{unres}} \in\mathbb{R}^{m\times (N-m)}$ is the $m\times (N-m)$ sub-matrix of $B$, obtained by selecting the rows of $B$ corresponding to the resolved variables and the columns of $B$ corresponding to the unresolved variables. The same logic applies to the other sub-matrices appearing in Eq.~(\ref{eq:Q_def}). Moreover, $e^{Bt}\in \mathbb{R}^{N \times N}$ is the matrix exponential, $e^{Bt}=\sum_{k=0}^\infty \frac{t^kB^k}{k!}$.
The covariance matrix $\Sigma(t)$ is governed by the Lyapunov equation,
\begin{eqnarray}
    \frac{d}{dt}\Sigma = B\Sigma^T + \Sigma B^T + \sigma\sigma^T,
\end{eqnarray}
whose explicit solution is given by,
\begin{eqnarray}
    \Sigma(t)=e^{Bt}\Sigma(0)e^{B^Tt} + \int_0^te^{B(t-s)}\sigma\sigma^Te^{B^T(t-s)}ds.
\end{eqnarray}
Finally, the solution of Eq.~(\ref{eq:evolution_error sde}) is given by,
\begin{eqnarray}\label{eq:error analytical sde}
    e(t) = -\int_0^t \Phi(t,\tau) \overline{F}(\tau)  d\tau,
\end{eqnarray}
where we defined the time-ordered exponential,
\begin{eqnarray}
    \Phi(t,\tau) = \mathcal{T}_{exp}\left(  \int_\tau^t Q(s)ds \right).
\end{eqnarray}
Equation~(\ref{eq:error short}) is therefore still valid in this stochastic setting provided Eq.~(\ref{eq:Q_def_sde}) is used for the definition of $Q(t)$.

\section{Numerical details}\label{app:numerical details}
\subsection{Initial covariance matrix}
\noindent The initial covariance matrix used in the numerical experiments presented in Sec \ref{sec:linear ode gaussian} and \ref{sec:linear sde} is given by,
\[
\Sigma_0 =
\begin{pmatrix}
0.3025 & 0.1500 & 0.1200 & 0.0550 & 0.0200 & 0.0000 & 0.0200 & 0.0550 & 0.1200 & 0.2050 \\
0.1500 & 0.3025 & 0.1500 & 0.1200 & 0.0550 & 0.0200 & 0.0000 & 0.0200 & 0.0550 & 0.1200 \\
0.1200 & 0.1500 & 0.3025 & 0.1500 & 0.1200 & 0.0550 & 0.0200 & 0.0000 & 0.0200 & 0.0550 \\
0.0550 & 0.1200 & 0.1500 & 0.3025 & 0.1500 & 0.1200 & 0.0550 & 0.0200 & 0.0000 & 0.0200 \\
0.0200 & 0.0550 & 0.1200 & 0.1500 & 0.3025 & 0.1500 & 0.1200 & 0.0550 & 0.0200 & 0.0000 \\
0.0000 & 0.0200 & 0.0550 & 0.1200 & 0.1500 & 0.3025 & 0.1500 & 0.1200 & 0.0550 & 0.0200 \\
0.0200 & 0.0000 & 0.0200 & 0.0550 & 0.1200 & 0.1500 & 0.3025 & 0.1500 & 0.1200 & 0.0550 \\
0.0550 & 0.0200 & 0.0000 & 0.0200 & 0.0550 & 0.1200 & 0.1500 & 0.3025 & 0.1500 & 0.1200 \\
0.1200 & 0.0550 & 0.0200 & 0.0000 & 0.0200 & 0.0550 & 0.1200 & 0.1500 & 0.3025 & 0.1500 \\
0.2050 & 0.1200 & 0.0550 & 0.0200 & 0.0000 & 0.0200 & 0.0550 & 0.1200 & 0.1500 & 0.3025
\end{pmatrix}
\]

\subsection{Initial bimodal probability measure and its PC expansion}\label{app:bimodal}
\noindent Given an arbitrary initial probability measure, computing the corresponding PC coefficients can be a complicated task. 
To approximate the PC coefficients of the probability measure $\nu$ in Eq.~(\ref{eq:bimodal measure}), we employ a data-driven approach. 
First, we train a residual neural network with a Sinkhorn loss to learn a transport map between a Gaussian distribution and our target mixture of Gaussian distributions \cite{feydy2019interpolating,feydy2019geomloss}. 
Next, this enables us to fit this transport map with a truncated PC expansion (Hermite polynomials up to total order 3). 
The PC expansion coefficients are solution of 
the following optimization problem,  
\sloppy
\begin{eqnarray}
    c^* = \operatorname*{arg\,min}_{c} \mathcal{L}(c),
\end{eqnarray}
with $c=(c_\alpha), c_\alpha\in\mathbb{R}^2, \alpha\in\mathcal{I}_3^2$
and the loss function,
\begin{eqnarray}
    \mathcal{L} = \lambda\, \hat{d}_{L^2}\big(X_{\text{NN}}(\xi),\, M_0(\xi; c)\big)
     + \lambda_b\, \hat{d}_{L^2}\big(b(X_{\text{NN}}(\xi)),\, b(M_0(\xi; c)\big), \\
    \hat{d}_{L^2}\big((X_i, Y_i)\big) := \left( \frac{1}{N_s} \sum_{i=1}^{N_s} \| X(\xi_i) - Y(\xi_i) \|^2 \right)^{1/2}, \nonumber  \\
    \xi_i \sim \mathcal{N}(0, I_N), \quad i = 1, \dots, N_s, \nonumber 
\end{eqnarray}
where $X(\xi), Y(\xi): \mathbb{R}^2\rightarrow\mathbb{R}^2$.
The map $X_{NN}(\xi)$ is the neural network trained in the first step. $M_0(\xi; c)$ is the truncated PC expansion whose coefficients are the arguments of our optimization problem. 
The optimization problem hence amounts to fit the distribution of $X\in\mathbb{R}^2$ with the truncated PC expansion while penalizing configuration incompatible with the distribution of the drift $b\in\mathbb{R}^2$. The relative importance of both distributions is controlled through the weights $\lambda,\lambda_b\geq0$. Accurate statistics of $b$ are indeed primordial for the reduced dynamics.

\subsection{Solution algorithm}
\begin{algorithm}[H]
\caption{Surrogate trajectories along probability flows}
\label{alg:solution algo}
\begin{algorithmic}[1]
\Require System dimension $N\in\mathbb{N}^+$, number of resolved variables $m<N$, polynomial drift coefficient $b\in\mathbb{R}^N$, constant noise coefficient matrix $\sigma\in\mathbb{R}^{N\times N}$, initial full-order probability measure $\mu_0$, initial condition for the resolved variables $\hat{X}_0\in \mathbb{R}^m$, time step $\mathrm{dt}>0$, number of samples $N_s\in\mathbb{N}^+$, PC expansion truncation total order $p\in\mathbb{N}^+$, conditional expectation polynomial expansion truncation total order $r\in\mathbb{N}^+$, quadrature points $\{q_k\}_{k\in K}$

\Statex 

\Statex \textbf{Offline Phase: probability flow}
\State Fit initial PC coefficients $c_j(0), \forall j\in \mathcal{I}_p^d$ to $\mu_0$

\For{$\tau = 0$ to $n_\tau$}
    \State Gaussian integrals: $I^{(1)}_{j,i}\leftarrow\mathbb{E}[b_iH_j]$ with analytical expressions
    \State Inverse of Jacobian matrix $(\frac{\partial M}{\partial \xi})^{-1}$ of the truncated PC expansion at quadrature points $\{q_k\}_{k\in K}$
    \State Gaussian integrals: $I^{D}_{j,i}\leftarrow\frac{1}{2} \mathbb{E} [ 
    \sigma^2 
    \sum_{l=1}^N\frac{\partial M_{t,l}^{-1}}{\partial \xi_i} 
    \frac{\partial H_j}{\partial \xi_l} 
    ]$ with Gauss-Hermite quadrature
    \State Euler increment: $c_{j,i}((\tau+1)\mathrm{dt}) \leftarrow c_{j,i}(\tau\mathrm{dt}) + I^{(1)}_{j,i} \mathrm{d}t + I^{D}_{j,i}\mathrm{d}t$
\EndFor
\vspace{0.2cm}
\Statex \textbf{Online Phase: surrogate trajectories}
\State $\hat{\sigma}_{i,j}=\sigma_{i,j}, \: i,j=\{1,\dots,m\}$
\For{$s=1$ to $N_s$}
    \State Initial condition: $\hat{X}\leftarrow\hat{X}_0$
    \For{$T = 0$ to  $n_\tau$}
        \State Gaussian integrals: $I^{(2)}$ and $I^{(3)}$ with analytical expressions
        \State Solve linear system: $\sum_{l\in\mathcal{J}_r^m}I^{(2)}_{k,l}\ \Bar{c}_{\alpha,j}(t) = I^{(3)}_{l,j}, \ \forall  j=\{1,..m\}, \ k\in \mathcal{J}_r^m$
        \State Polynomial expansion: $\mathbb{E}[\hat{b} \mid \hat{X}]\leftarrow\sum_{\beta\in\mathcal{J}_r^m}\Bar{c}_\alpha\Bar{H}_{\alpha}(\hat{X})$ 
        \State Wiener increment: $dW \sim \sqrt{dt}\ \mathcal{N}(0,I_m)$
        \State Euler-Maruyama increment: $\mathrm{d}\hat{X} \leftarrow \mathbb{E}[b \mid \hat{X}]\,\mathrm{dt} + \hat{\sigma}\,\mathrm{d}\hat{W}$
        \State $\hat{X}^{(s)} \leftarrow \hat{X}^{(s)} + \mathrm{d}\hat{X}$
    \EndFor
\EndFor
\State Averages: $\mathbb{E}[\hat{X}_{t=Tdt}\mid\hat{X}_0]\leftarrow \frac{1}{N_s}\sum_{s=1}^{N_s}X^{(s)}_{Tdt}, \ T=\{0,\dots,n_\tau\}$ 
\end{algorithmic}
\end{algorithm}

\newpage

\bibliographystyle{siam}
\bibliography{ref.bib}

@article{gyongy1986mimicking,
  title={Mimicking the one-dimensional marginal distributions of processes having an It{\^o} differential},
  author={Gy{\"o}ngy, Istv{\'a}n},
  journal={Probability theory and related fields},
  volume={71},
  number={4},
  pages={501--516},
  year={1986},
  publisher={Springer}
}

@article{brunick2013mimicking,
  title={Mimicking an It{\^o} process by a solution of a stochastic differential equation},
  author={Gerard P. Brunick and Steven E. Shreve},
  journal={Annals of Applied Probability},
  year={2010},
  volume={23},
  pages={1584-1628},
  url={https://api.semanticscholar.org/CorpusID:11597682}
}

@book{Villani2009,
  author    = {Cédric Villani},
  title     = {Optimal Transport: Old and New},
  publisher = {Springer},
  year      = {2009},
}

@article{Givon2005orthogdyn,
author = {Givon, Dror and Kupferman, R. and Hald, O. H.},
title = {Existence proof for orthogonal dynamics and the Mori-Zwanzig formalism},
journal = {Israel Journal of Mathematics},
volume = {145},
number = {1},
pages = {221-241},
year = {2005},
doi = { https://doi.org/10.1007/bf02786691}
}

@article{WANG2020109402,
title = {Recurrent neural network closure of parametric POD-Galerkin reduced-order models based on the Mori-Zwanzig formalism},
journal = {Journal of Computational Physics},
volume = {410},
pages = {109402},
year = {2020},
issn = {0021-9991},
doi = {https://doi.org/10.1016/j.jcp.2020.109402},
url = {https://www.sciencedirect.com/science/article/pii/S0021999120301765},
author = {Qian Wang and Nicolò Ripamonti and Jan S. Hesthaven}
}

@article{LIN2021109864,
title = {Data-driven model reduction, Wiener projections, and the Koopman-Mori-Zwanzig formalism},
journal = {Journal of Computational Physics},
volume = {424},
pages = {109864},
year = {2021},
issn = {0021-9991},
doi = {https://doi.org/10.1016/j.jcp.2020.109864},
url = {https://www.sciencedirect.com/science/article/pii/S0021999120306380},
author = {Kevin K. Lin and Fei Lu}
}

@article{
chorin2000MZ,
author = {Alexandre J. Chorin  and Ole H. Hald  and Raz Kupferman },
title = {Optimal prediction and the Mori–Zwanzig representation of irreversible processes},
journal = {Proceedings of the National Academy of Sciences},
volume = {97},
number = {7},
pages = {2968-2973},
year = {2000},
doi = {10.1073/pnas.97.7.2968},
URL = {https://www.pnas.org/doi/abs/10.1073/pnas.97.7.2968},
eprint = {https://www.pnas.org/doi/pdf/10.1073/pnas.97.7.2968}
}

@book{Evans_Morriss_2008, place={Cambridge}, edition={2}, title={Statistical Mechanics of Nonequilibrium Liquids}, publisher={Cambridge University Press}, author={Evans, Denis J. and Morriss, Gary}, year={2008}}

@book{grabert1982proj,
title = {Projection Operator Techniques in Nonequilibrium Statistical Mechanics},
edition={1},
place={Berlin},
publisher={Springer},
author={Grabert, Herman},
year = {1982},
}

@article{gouasmi2017memory,
title = {A priori estimation of memory effects in reduced-order models of nonlinear systems using the Mori-Zwanzig formalism},
author = {Gouasmi, Ayoub  and Parish, Eric J.  and Duraisamy, Karthik },
journal = {Proceedings of the Royal Society A: Mathematical, Physical and Engineering Sciences},
volume = {473},
number = {2205},
pages = {20170385},
year = {2017},
doi = {10.1098/rspa.2017.0385},
URL = {https://royalsocietypublishing.org/doi/abs/10.1098/rspa.2017.0385},
eprint = {https://royalsocietypublishing.org/doi/pdf/10.1098/rspa.2017.0385}
}

@article{CHORIN2009optimal_prediction_memory,
title = {Optimal prediction with memory},
journal = {Physica D: Nonlinear Phenomena},
volume = {166},
number = {3},
pages = {239-257},
year = {2002},
issn = {0167-2789},
doi = {https://doi.org/10.1016/S0167-2789(02)00446-3},
url = {https://www.sciencedirect.com/science/article/pii/S0167278902004463},
author = {Alexandre J. Chorin and Ole H. Hald and Raz Kupferman}
}

@misc{jeong2024diffgledifferentiablecoarsegraineddynamics,
      title={DiffGLE: Differentiable Coarse-Grained Dynamics using Generalized Langevin Equation}, 
      author={Jinu Jeong and Ishan Nadkarni and Narayana. R. Aluru},
      year={2024},
      eprint={2410.08424},
      archivePrefix={arXiv},
      primaryClass={cond-mat.soft},
      url={https://arxiv.org/abs/2410.08424}, 
}

@article{Vrugt_2020,
doi = {10.1088/1361-6404/ab8e28},
url = {https://dx.doi.org/10.1088/1361-6404/ab8e28},
year = {2020},
month = {jun},
publisher = {IOP Publishing},
volume = {41},
number = {4},
pages = {045101},
author = {Michael te Vrugt and Raphael Wittkowski},
title = {Projection operators in statistical mechanics: a pedagogical approach},
journal = {European Journal of Physics}
}

@book{_ttinger_2005,
	doi = {10.1002/0471727903},
	url = {https://doi.org/10.1002\%2F0471727903},
	year = 2005,
	month = {},
	publisher = {John Wiley {\&} Sons, Inc.},
	author = {H.C. Öttinger},
	title = {Beyond Equilibrium Thermodynamics}
}

@book{Süli_Mayers_2003, place={Cambridge}, title={An Introduction to Numerical Analysis}, publisher={Cambridge University Press}, author={Süli, Endre and Mayers, David F.}, year={2003}}

@article{gorjichaos,
    author = {Gorji, Hossein},
    title = {Logarithmic Gradient Transformation and Chaos Expansion of Itô Processes},
    pages = {215-231},
    volume = {10},
    journal = {Communications in Mathematics and Statistics},
    year = {2022},
    doi = {10.1007/s40304-020-00219-2},
    URL = {https://doi.org/10.1007/s40304-020-00219-2},
}

@inproceedings{feydy2019interpolating,
  title={Interpolating between optimal transport and mmd using sinkhorn divergences},
  author={Feydy, Jean and S{\'e}journ{\'e}, Thibault and Vialard, Fran{\c{c}}ois-Xavier and Amari, Shun-ichi and Trouv{\'e}, Alain and Peyr{\'e}, Gabriel},
  booktitle={The 22nd international conference on artificial intelligence and statistics},
  pages={2681--2690},
  year={2019},
  organization={PMLR}
}

@misc{feydy2019geomloss,
  author = {Jean Feydy and others},
  title = {GeomLoss: A Python library for computing kernel losses and Sinkhorn divergences},
  year = {2019},
  howpublished = {\url{https://www.kernel-operations.io/geomloss/}},
  note = {Version accessed August 2025}
}

@incollection{hesthaven2015model,
  title={Model Reduction},
  author={Hesthaven, Jan S},
  booktitle={Encyclopedia of Applied and Computational Mathematics},
  pages={923--925},
  year={2015},
  publisher={Springer}
}

@article{BAI20029,
title = {Krylov subspace techniques for reduced-order modeling of large-scale dynamical systems},
journal = {Applied Numerical Mathematics},
volume = {43},
number = {1},
pages = {9-44},
year = {2002},
note = {19th Dundee Biennial Conference on Numerical Analysis},
issn = {0168-9274},
doi = {https://doi.org/10.1016/S0168-9274(02)00116-2},
url = {https://www.sciencedirect.com/science/article/pii/S0168927402001162},
author = {Zhaojun Bai}
}

@article{POD,
 ISSN = {00113891},
 URL = {http://www.jstor.org/stable/24103957},
 author = {Anindya Chatterjee},
 journal = {Current Science},
 number = {7},
 pages = {808--817},
 publisher = {Temporary Publisher},
 title = {An introduction to the proper orthogonal decomposition},
 urldate = {2025-08-12},
 volume = {78},
 year = {2000}
}

@article{quarteroni,
author = {Quarteroni, Alfio and Rozza, Gianluigi and Manzoni, Andrea},
year = {2011},
month = {12},
pages = {},
title = {Certified Reduced Basis Approximation for Parametrized Partial Differential Equations and Applications},
volume = {1},
journal = {Journal of Mathematics in Industry},
doi = {10.1186/2190-5983-1-3}
}

@article{rathinam,
author = {Rathinam, Muruhan and Petzold, Linda},
year = {2003},
month = {01},
pages = {1893-1925},
title = {A New Look at Proper Orthogonal Decomposition},
volume = {41},
journal = {SIAM J. Numerical Analysis},
doi = {10.1137/S0036142901389049}
}

@article{rowley,
author = {Rowley, Clarence},
year = {2005},
month = {03},
pages = {},
title = {MODEL REDUCTION FOR FLUIDS, USING BALANCED PROPER ORTHOGONAL DECOMPOSITION},
volume = {15},
journal = {International Journal of Bifurcation and Chaos},
doi = {10.1142/S0218127405012429}
}

@article{SCHMID_2010, title={Dynamic mode decomposition of numerical and experimental data}, volume={656}, DOI={10.1017/S0022112010001217}, journal={Journal of Fluid Mechanics}, author={SCHMID, PETER J.}, year={2010}, pages={5–28}}

@book{gallavotti2013statistical,
  title={Statistical mechanics: A short treatise},
  author={Gallavotti, Giovanni},
  year={2013},
  publisher={Springer Science \& Business Media}
}

@book{ellis2012entropy,
  title={Entropy, large deviations, and statistical mechanics},
  author={Ellis, Richard S},
  volume={271},
  year={2012},
  publisher={Springer Science \& Business Media}
}

@incollection{risken1989fokker,
  title={Fokker-planck equation},
  author={Risken, Hannes},
  booktitle={The Fokker-Planck equation: methods of solution and applications},
  pages={63--95},
  year={1989},
  publisher={Springer}
}

@incollection{kloeden2011stochastic,
  title={Stochastic Differential Equations},
  author={Kloeden, Peter E},
  booktitle={International Encyclopedia of Statistical Science},
  pages={1520--1521},
  year={2011},
  publisher={Springer}
}

@article{rogers1996diffusions,
  title={Diffusions, Markov processes and martingales: vol, 1. foundations},
  author={Rogers, Leonard CG and Williams, David},
  journal={Journal of the Royal Statistical Society-Series A Statistics in Society},
  volume={159},
  number={2},
  pages={343},
  year={1996},
  publisher={London: Royal Statistical Society, 1988-}
}

@article{Parish,
  title = {Non-Markovian closure models for large eddy simulations using the Mori-Zwanzig formalism},
  author = {Parish, Eric J. and Duraisamy, Karthik},
  journal = {Phys. Rev. Fluids},
  volume = {2},
  issue = {1},
  pages = {014604},
  numpages = {28},
  year = {2017},
  month = {Jan},
  publisher = {American Physical Society},
  doi = {10.1103/PhysRevFluids.2.014604},
  url = {https://link.aps.org/doi/10.1103/PhysRevFluids.2.014604}
}

@article{budhiraja2019analysis,
  title={Analysis and approximation of rare events},
  author={Budhiraja, Amarjit and Dupuis, Paul},
  journal={Representations and Weak Convergence Methods. Series Prob. Theory and Stoch. Modelling},
  volume={94},
  pages={8},
  year={2019},
  publisher={Springer}
}

@article{schwerdtner2024greedy,
  title={Greedy construction of quadratic manifolds for nonlinear dimensionality reduction and nonlinear model reduction},
  author={Schwerdtner, Paul and Peherstorfer, Benjamin},
  journal={arXiv preprint arXiv:2403.06732},
  year={2024}
}

@inproceedings{otto2024role,
  title={On the role of the projection fiber for modeling transient nonlinear dynamics},
  author={Otto, Samuel E and Kutz, Nathan and Brunton, Steven L},
  booktitle={76th Annual Meeting of the Division of Fluid Dynamics},
  year={2023},
  organization={APS}
}

@article{buchfink2024model,
  title={Model reduction on manifolds: A differential geometric framework},
  author={Buchfink, Patrick and Glas, Silke and Haasdonk, Bernard and Unger, Benjamin},
  journal={Physica D: Nonlinear Phenomena},
  volume={468},
  pages={134299},
  year={2024},
  publisher={Elsevier}
}

@inproceedings{qian2019transform,
  title={Transform \& learn: A data-driven approach to nonlinear model reduction},
  author={Qian, Elizabeth and Kramer, Boris and Marques, Alexandre N and Willcox, Karen E},
  booktitle={AIAA Aviation 2019 Forum},
  pages={3707},
  year={2019}
}

@article{qian2022reduced,
  title={Reduced operator inference for nonlinear partial differential equations},
  author={Qian, Elizabeth and Farcas, Ionut-Gabriel and Willcox, Karen},
  journal={SIAM Journal on Scientific Computing},
  volume={44},
  number={4},
  pages={A1934--A1959},
  year={2022},
  publisher={SIAM}
}

@article{willcox2024distributed,
  title={Distributed computing for physics-based data-driven reduced modeling at scale: Application to a rotating detonation rocket engine},
  author={Willcox, Karen E},
  journal={arXiv preprint arXiv:2407.09994},
  year={2024}
}

@article{lee2020model,
  title={Model reduction of dynamical systems on nonlinear manifolds using deep convolutional autoencoders},
  author={Lee, Kookjin and Carlberg, Kevin T},
  journal={Journal of Computational Physics},
  volume={404},
  pages={108973},
  year={2020},
  publisher={Elsevier}
}

\end{document}